\documentclass[journal]{IEEEtran}

\usepackage{amsfonts}
\usepackage{amssymb}
\usepackage{amsthm}
\usepackage{amsmath,amsfonts,amssymb}
\usepackage[dvips]{graphicx}
\usepackage{subfigure}
\usepackage{verbatim}
\usepackage{setspace}
\usepackage{bm}
\usepackage{algorithmic} 
\usepackage[ruled,vlined]{algorithm2e}
\usepackage{cite}

\usepackage{changepage}
\usepackage{pdfpages}
\usepackage{color}
\newtheorem{theorem}{Theorem}

\newcommand*{\e}{\mathop{}\!\mathrm{e}}
\newcommand*{\jj}{\mathop{}\!\mathrm{j}}
\allowdisplaybreaks

\newcommand{\figwidth}{7.8}
\IEEEoverridecommandlockouts

\begin{document}
\title{Towed Movable Antenna (ToMA) Array \\for Ultra Secure Airborne Communications}
\author{Lipeng Zhu,~\IEEEmembership{Member,~IEEE,}
		Haobin Mao,~\IEEEmembership{Graduate Student Member,~IEEE,}
		Wenyan Ma,~\IEEEmembership{Graduate Student Member,~IEEE,}
		Zhenyu Xiao,~\IEEEmembership{Senior Member,~IEEE,}
		Jun Zhang,
		and Rui Zhang,~\IEEEmembership{Fellow,~IEEE}
	\vspace{-0.5 cm}
	\thanks{L. Zhu and W. Ma are with the Department of Electrical and Computer Engineering, National University of Singapore, Singapore 117583 (e-mail: zhulp@nus.edu.sg, wenyan@u.nus.edu).}
	\thanks{H. Mao and Z. Xiao are with the School of Electronic and Information Engineering, Beihang University, Beijing, China 100191 (e-mail: maohaobin@buaa.deu.cn, xiaozy@buaa.edu.cn).}
	\thanks{J. Zhang is with the State Key Laboratory of CNS/ATM \& MIIT Key Laboratory of Complex-field Intelligent Sensing, Beijing Institute of Technology, Beijing 100081, China (e-mail: zhjun@bit.edu.cn). He is also with the School of Electronic and Information Engineering, Beihang University, Beijing 100191, China.}
	\thanks{R. Zhang is with School of Science and Engineering, Shenzhen Research Institute of Big Data, The Chinese University of Hong Kong, Shenzhen, Guangdong 518172, China (e-mail: rzhang@cuhk.edu.cn). He is also with the Department of Electrical and Computer Engineering, National University of Singapore, Singapore 117583 (e-mail: elezhang@nus.edu.sg).}
}

\maketitle


\begin{abstract}	
	This paper proposes a novel towed movable antenna (ToMA) array architecture to enhance the physical layer security of airborne communication systems. Unlike conventional onboard arrays with fixed-position antennas (FPAs), the ToMA array employs multiple subarrays mounted on flexible cables and towed by distributed drones, enabling agile deployment in three-dimensional (3D) space surrounding the central aircraft. This design significantly enlarges the effective array aperture and allows dynamic geometry reconfiguration, offering superior spatial resolution and beamforming flexibility. We consider a secure transmission scenario where an airborne transmitter communicates with multiple legitimate users in the presence of potential eavesdroppers. To ensure security, zero-forcing beamforming is employed to nullify signal leakage toward eavesdroppers. Based on the statistical distributions of locations of users and eavesdroppers, the antenna position vector (APV) of the ToMA array is optimized to maximize the users' ergodic achievable rate. Analytical results for the case of a single user and a single eavesdropper reveal the optimal APV structure that minimizes their channel correlation. For the general multiuser scenario, we develop a low-complexity alternating optimization algorithm by leveraging Riemannian manifold optimization. Simulation results confirm that the proposed ToMA array achieves significant performance gains over conventional onboard FPA arrays, especially in scenarios where eavesdroppers are closely located to users under line-of-sight (LoS)-dominant channels.
\end{abstract}

%
\IEEEpeerreviewmaketitle

\section{Introduction}
\subsection{Background}
\IEEEPARstart{A}{irborne} communication has gained increasing attention in recent years due to its potential to provide agile and wide-area wireless connectivity \cite{cao2018airborne,zeng2019access,Baltaci2021aerial}. Leveraging aerial platforms such as unmanned aerial vehicles (UAVs) and high-altitude platforms (HAPs), airborne networks have demonstrated significant advantages in a variety of civilian applications, including emergency response, environmental monitoring, and temporary broadband coverage in underserved areas. Beyond civilian scenarios, airborne communication also plays a vital role in military and defense systems \cite{He2017UAVsecure,Wang2019UAVsecure}, where secure and resilient communication links are of paramount importance. In such mission-critical contexts, any leakage of transmitted information could lead to severe consequences. The open and broadcast nature of wireless channels, combined with the elevated and thus easily exposed positions of aerial platforms, makes them particularly vulnerable to hostile eavesdropping and jamming.

To address the vulnerability of wireless transmissions, physical layer security (PLS) has emerged as a promising paradigm that complements traditional cryptographic methods \cite{liu2017secure,chen2017secure}. By exploiting the spatial, temporal, and spectral characteristics of wireless channels, PLS techniques aim to enhance the confidentiality of wireless communications at the signal level. Representative approaches include artificial noise (AN) generation to confuse eavesdroppers, cooperative relaying or jamming, and opportunistic beamforming based on channel state information (CSI). Among them, transmit beamforming based on antenna arrays is particularly attractive in airborne systems due to the high probability of line-of-sight (LoS) links \cite{chen2017secure,xiao2022mmWaveUAV}. By concentrating signal energy toward legitimate users while spatially suppressing leakage to potential eavesdroppers, beamforming can significantly improve secrecy performance without additional power or bandwidth overhead \cite{xiao2022mmWaveUAV,zhang2025secureUAV,dong2021UAVsecure}.

However, the effectiveness of array beamforming for secure communication is fundamentally limited by the spatial correlation between the channels of legitimate users and hostile eavesdroppers \cite{zhang2025secureUAV,Mukherjee2011secure,Ng2014secure}. For example, when the channels of user and eavesdropper are nearly orthogonal, transmit beamforming can efficiently enhance the signal power received at users while minimizing information leakage to eavesdroppers. In contrast, when their channels are highly correlated, such as in scenarios where they are located in close proximity or aligned in similar directions, beamforming inevitably leads to a trade-off between signal enhancement and leakage suppression. This intrinsic bottleneck cannot be resolved by simply optimizing beamforming weights in conventional systems employing fixed-position antenna (FPA) arrays. To address this challenge, movable antenna (MA) technology has recently emerged as a promising solution \cite{zhu2022MAmodel,zhu2023MAMag,ma2022MAmimo}. By proactively adjusting the antenna positions, MA systems can reduce channel correlation between users and eavesdroppers, thereby enabling more effective secure transmission when combined with beamforming. Recent works have demonstrated the significant secrecy performance gains of MA-assisted systems compared to traditional FPA architectures \cite{hu2024secure,cheng2024secure,tang2024secure,Ding2024MAsecure,mei2024MAsecure,liu2025MAcovert,feng2024MAsecure,mao2025MAcovert}.

\subsection{Related Works}
The MA technology has a long and evolving development history \cite{zhu2025tutorial}. Over the past few decades, various antennas with mechanical movement or rotational capabilities have been proposed to enhance radiation efficiency or improve received signal quality \cite{balanis2007modern,Basbug2017design}. To the best of our knowledge, the earliest rigorous study on MA-aided wireless communication systems can be traced back to 2009 \cite{zhao2009singleMA}, where the authors characterized the spatial diversity of a single receive MA moving along a one-dimensional (1D) trajectory, based on a spatial-correlation channel model. More recently, the development of field-response channel models and antenna movement optimization frameworks has significantly advanced this research area \cite{zhu2022MAmodel,zhu2023MAMag,ma2022MAmimo}. This line of work is sometimes also referred to as fluid antenna in terms of flexible antenna positioning \cite{Xu2024FASmultiple,hu2024power,wu2024fluidMag}. A wide range of studies have highlighted the performance benefits of MA systems over their FPA counterparts, including received signal power improvement \cite{zhao2009singleMA,zhu2022MAmodel,mei2024movable}, effective interference suppression \cite{zhu2023MAMag,wang2024MAinterference,ning2024movable}, enhanced beamforming flexibility \cite{zhu2023MAarray,ma2024multi,wang2024flexible}, and capacity gains in multiple-input multiple-output (MIMO) and/or multiuser communication systems \cite{ma2022MAmimo,chen2023joint,zhu2023MAmultiuser,xiao2023multiuser,wu2024globallyMA,Feng2024MAweighted,Xu2024FASmultiple,zhang2024MAhybrid}. In parallel, the problem of acquiring accurate CSI for MA systems has been extensively investigated  \cite{ma2023MAestimation,xiao2023channel,zhang2024TensorCE,Jang2025MAchannel}, enabling the reconstruction of spatially continuous channel mapping between transmit and receive regions.

To further exploit the spatial degrees of freedom (DoFs), the six-dimensional MA (6DMA) system has been proposed \cite{shao20246DMA,shao2024discrete,shao20246DMANet,shao2025tutorial,pi20246DMAcoordi,ren20246DMAUAV}, allowing full control of both three-dimensional (3D) positions and 3D orientations of antennas. For lower-cost implementations, rotatable antennas that support orientation adjustments only have been explored as a simplified form of 6DMA \cite{zheng2025rotatable,zheng2025rotatableMag}. More recently, the pinching antenna concept has been introduced to enable flexible antenna positioning along pre-deployed 1D waveguides \cite{ding2024pinching,liu2025pinching,ouyang2025pinching}, while the extremely large-scale MA (XL-MA) architecture supports antenna movement on two-dimensional (2D) surfaces \cite{fu2025extremelyMA}. Both pinching antennas and XL-MAs leverage large-scale antenna position optimization to proactively establish LoS channels, reduce large-scale path loss, and suppress multiuser interference.

However, in airborne platforms such as UAVs and HAPs, the physical dimensions and payload capacities are inherently constrained \cite{zeng2019access,xiao2022mmWaveUAV,zhu2020UAVrelay}. These practical limitations not only restrict the number of deployable antennas but also confine the maximum array aperture for both conventional FPA arrays and existing MA arrays. Although MAs can flexibly adjust their positions within a spatial region, their movement is still limited by the structural envelope of the hosting platform. As a result, when legitimate users and eavesdroppers are located in similar directions or physically close to each other, the spatial resolution and beamforming flexibility of onboard antenna arrays become insufficient to distinguish between them. The resultant high channel correlation makes it difficult to effectively suppress information leakage, thereby severely degrading secrecy performance \cite{zhu2025tutorial}. These limitations highlight the need for a more flexible and scalable antenna architecture that can overcome the physical constraints of airborne platforms.

\begin{figure}[t]
	\begin{center}
		\includegraphics[width=\figwidth cm]{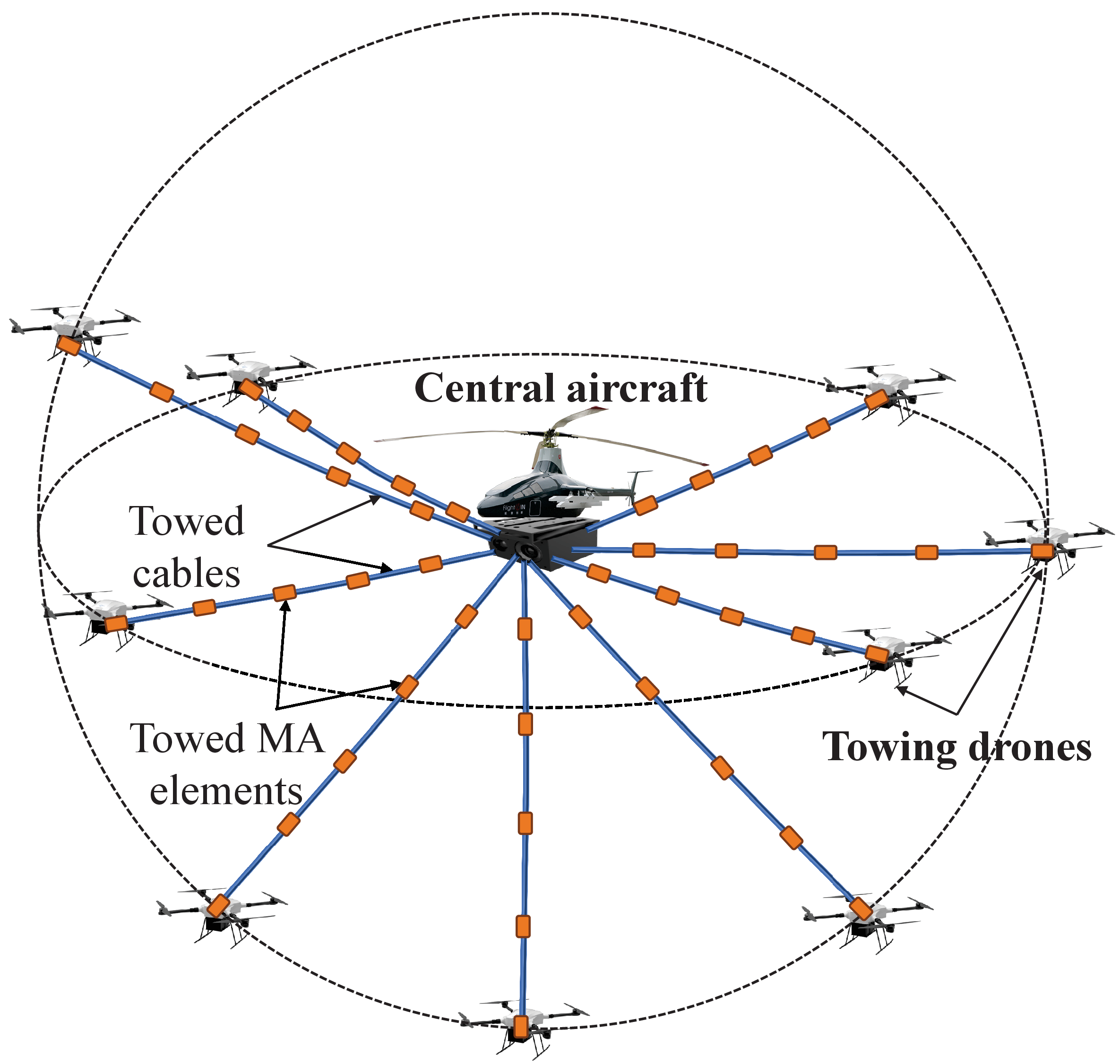}
		\caption{The architecture of the proposed ToMA array.}
		\label{fig:architecture}
	\end{center}
\end{figure}
\subsection{Contributions}

To address the limitations of existing MA architectures on airborne platforms, we propose in this paper a novel towed MA (ToMA) array for secure airborne communications. As illustrated in Fig. \ref{fig:architecture}, the central aircraft is equipped with multiple antenna subarrays mounted on flexible cables. Towed by distributed drones, these antenna subarrays can be flexibly deployed throughout the 3D space surrounding the central aircraft. This design offers two key advantages. First, the effective array aperture of the ToMA array can far exceed the physical dimensions of the airborne platform, allowing for ultra-large aperture array configurations that significantly improve spatial resolution. Second, the 3D positions of the towed subarrays can be dynamically reconfigured, enabling the array geometry to adapt in real time to the relative locations of users and eavesdroppers. These capabilities provide unprecedented flexibility in spatial beamforming and interference/leakage suppression, and thus lead to significant enhancements in PLS performance, especially in challenging environments where legitimate users and hostile eavesdroppers are closely spaced. Furthermore, these towed cables can be retracted when communication is not required. The drones and antenna subarrays can then be stored inside the central aircraft, enabling a compact and practical deployment.\footnote{The proposed ToMA array is also applicable to other size-constrained mobile platforms, such as satellites, terrestrial vehicles, naval vessels, submarines, etc.}

In this paper, we apply the proposed ToMA array to airborne secure communication systems, with the main contributions summarized as follows:
\begin{itemize}	
	\item We propose a novel ToMA array architecture for secure airborne communications. Specifically, we consider downlink transmission from the ToMA array to multiple legitimate users in the presence of multiple potential eavesdroppers. To ensure secure and covert communication, zero-forcing (ZF) beamforming is employed to completely eliminate signal leakage to eavesdroppers. Based on the statistical distributions of locations of users and eavesdroppers, the antenna position vector (APV) of the ToMA array is optimized to maximize the users' ergodic achievable rate under practical deployment constraints.
	\item To gain analytical insights, we investigate a special case involving only a single user and a single eavesdropper. We derive the minimum array response correlation that maximizes the user's achievable rate and reveal the optimal APV structure when the user and eavesdropper are closely located. The result highlights how the array geometry can be tailored to maximize the effective aperture in the critical angular region between the user and eavesdropper.
	\item For the general case with multiple users and eavesdroppers, we develop a low-complexity algorithm based on alternating optimization (AO) and Riemannian manifold optimization. The ergodic achievable rate is first approximated via Monte Carlo simulations. Then, the position of each towed subarray is alternately updated while fixing the others. In particular, each subarray’s position lies on a spherical manifold, enabling efficient optimization via Riemannian gradient methods.
	\item Extensive simulations validate that the proposed ToMA array significantly outperforms conventional onboard FPA arrays in secure communication performance. It is shown that increasing the cable length enlarges the effective aperture and improves rate performance, even approaching the theoretical upper bound. Furthermore, dynamic reconfiguration of antenna positions allows the ToMA array to flexibly adapt to various 3D user/eavesdropper distributions. The performance advantage is especially prominent in LoS-dominant scenarios when users and eavesdroppers are closely spaced.
\end{itemize}

The remainder of this paper is organized as follows. Section II introduces the ToMA array architecture and the corresponding system model. In Section III, analytical results are derived for the special case of a single user and a single eavesdropper. Section IV presents the proposed antenna position optimization algorithm. Simulation results are provided in Section V to evaluate the performance of the proposed scheme, and Section VI concludes the paper.

\textit{Notation}: $a$, $\mathbf{a}$, $\mathbf{A}$, and $\mathcal{A}$ denote a scalar, a vector, a matrix, and a set, respectively. $(\cdot)^{\rm{T}}$ and $(\cdot)^{\rm{H}}$ denote the transpose and conjugate transpose, respectively. $[\mathbf{a}]_{n}$ denotes the $n$-th entry of vector $\mathbf{a}$. $[\mathbf{A}]_{i,j}$ and $[\mathbf{A}]_{:,j}$ denote the entry in row $i$ and column $j$ and the $j$-th column vector of matrix $\mathbf{A}$, respectively. $\mathbb{Z}^{M \times N}$, $\mathbb{R}^{M \times N}$, and $\mathbb{C}^{M \times N}$ represent the sets of integer, real, and complex matrices of dimension $M \times N$, respectively. $\|\mathbf{a}\|$ denotes the 2-norm of vector $\mathbf{a}$, and $\|\mathbf{A}\|_{\mathrm{F}}$ represents the Frobenius norm of matrix $\mathbf{A}$. $\mathbf{I}_{M}$ denotes the $M$-dimensional identity matrix. $\mathbb{E}\{\cdot\}$ is the expectation of a random variable.

\begin{figure}[t]
	\begin{center}
		\includegraphics[width=8.8 cm]{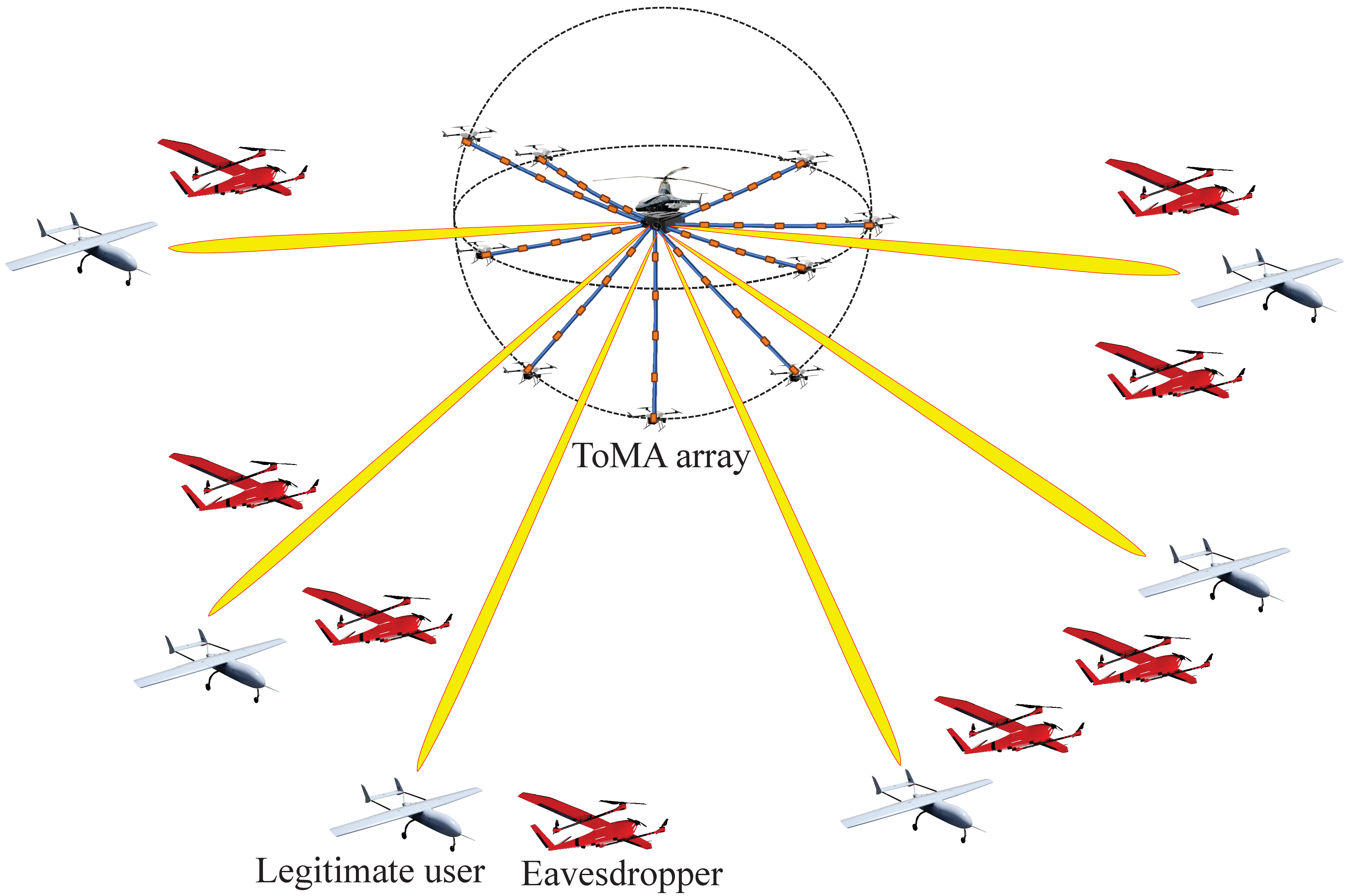}
		\caption{The considered ToMA array enabled airborne secure communication system.}
		\label{fig:system}
	\end{center}
\end{figure}

\section{System Model}
As shown in Fig. \ref{fig:system}, we consider an airborne secure communication scenario. The central aircraft is assumed to be hovering at a given position and transmits secrecy messages to legitimate users, while several hostile eavesdroppers aim to intercept these information. The numbers of users and eavesdroppers are denoted as $K$ and $I$, respectively. For simplicity, we assume that the users and eavesdroppers are each equipped with a single antenna.\footnote{In the case where the users or eavesdroppers are equipped with multiple antennas, a directional beam should be steered to the central aircraft for maximizing the received signal power, which can be equivalently regarded as a single directional antenna.} 
To ensure secure transmission, the central aircraft is equipped with a ToMA array capable of dynamically reconfiguring its aperture and geometry. Specifically, $M$ flexible cables are attached to the aircraft, with $N$ antenna elements uniformly distributed along each cable. Each cable is towed by an individual drone and can be extended outward to form a uniform linear array (ULA). When communication is not required, the cables can be retracted, and both the drones and subarrays can be conveniently stored within the aircraft. This flexible architecture allows the ToMA array to adapt to diverse operational environments and transmission requirements.

We establish a 3D Cartesian coordinate system centered at the central aircraft to describe the positions of antennas, users, and eavesdroppers. Specifically, the coordinates of user $k$ and eavesdropper $i$ are denoted as $\mathbf{r}_{\mathrm{u},k}=[X_{k}, Y_{k}, Z_{k}]^{\mathrm{T}}$, $1 \leq k \leq K$, and $\mathbf{r}_{\mathrm{e},i}=[\tilde{X}_{i}, \tilde{Y}_{i}, \tilde{Z}_{i}]^{\mathrm{T}}$, $1 \leq i \leq I$, respectively. The coordinate of the $m$-th towing drone is denoted as $\mathbf{t}_{m}=[x_{m}, y_{m}, z_{m}]^{\mathrm{T}}$, $1 \leq m \leq M$. To prevent the cable from sagging, the distance between each drone to the central aircraft is maximized to be $\left\|\mathbf{t}_{m}\right\|=L$. Due to the uniform distribution of antenna elements, the coordinate of the $n$-th antenna element on the $m$-th towed cable is given by $\mathbf{t}_{m,n}=\frac{n}{N} \mathbf{t}_{m}$, $1 \leq n \leq N$. For ease of exposition, the positions of all drones are collected into the antenna position vector (APV) $\tilde{\mathbf{t}}=[\mathbf{t}_{1}^{\mathrm{T}}, \mathbf{t}_{2}^{\mathrm{T}}, \dots, \mathbf{t}_{M}^{\mathrm{T}}]^{\mathrm{T}} \in \mathbb{C}^{3M \times 1}$, which uniquely determines the array geometry.

Due to the rare scatterers and obstacles, the air-to-air channels are dominated by LoS paths. Since the aperture size of the ToMA array is extremely large, typically on the order of a few meters to several tens of meters, the users and eavesdroppers may be located within the near-field region of the ToMA array. Thus, we adopt the uniform spherical wave model to characterize the array response vector, which is given by
\begin{equation}\label{eq_array}
	\begin{aligned}
		\mathbf{a}(\tilde{\mathbf{t}}, \mathbf{r}) = \Big{[}\e^{\jj \frac{2\pi}{\lambda} \|\mathbf{t}_{m,n}-\mathbf{r}\|}\Big{]}^{\mathrm{T}}_{1 \leq m \leq M, 1 \leq n \leq N} \in \mathbb{C}^{MN \times 1}.
	\end{aligned}
\end{equation}

Then, the channel vectors for user $k$ and eavesdropper $i$ are respectively given by\footnote{The proposed ToMA architecture can be readily extended to configurations with non-uniform inter-antenna spacing along each cable and unequal lengths among multiple cables. Besides, each antenna element is assumed to exhibit a quasi-omnidirectional radiation pattern, and potential blockage from the central aircraft is ignored. These simplifying assumptions are adopted to focus on the core system design, while a more comprehensive treatment is left for future work.}
\begin{subequations}\label{eq_channel}
	\begin{align}
		&\mathbf{h}_{k}(\tilde{\mathbf{t}}) = \alpha_{k} \mathbf{a}(\tilde{\mathbf{t}}, \mathbf{r}_{\mathrm{u},k}) \in \mathbb{C}^{MN \times 1},~1 \leq k \leq K,\\
		&\mathbf{g}_{i}(\tilde{\mathbf{t}}) = \beta_{i} \mathbf{a}(\tilde{\mathbf{t}}, \mathbf{r}_{\mathrm{e},i}) \in \mathbb{C}^{MN \times 1},~1 \leq i \leq I,
	\end{align}
\end{subequations}
where $\alpha_{k}$ is the path gain of user $k$ and $\beta_{i}$ denotes the path gain of eavesdropper $i$. Note that the channel model in \eqref{eq_channel} can be readily extended to Rician fading with random non-LoS (NLoS) components. The performance comparison under different Rician factors will be presented via simulations in Section V. The channel matrices for all $K$ users and all $I$ eavesdroppers are denoted as $\mathbf{H}(\tilde{\mathbf{t}}) = [\mathbf{h}_{1}(\tilde{\mathbf{t}}), \mathbf{h}_{2}(\tilde{\mathbf{t}}), \dots, \mathbf{h}_{K}(\tilde{\mathbf{t}})] \in \mathbb{C}^{MN \times K}$ and $\mathbf{G}(\tilde{\mathbf{t}}) = [\mathbf{g}_{1}(\tilde{\mathbf{t}}), \mathbf{g}_{2}(\tilde{\mathbf{t}}), \dots, \mathbf{g}_{I}(\tilde{\mathbf{t}})] \in \mathbb{C}^{MN \times I}$.

Denoting the transmit signal as $\mathbf{s} \in \mathbb{C}^{K \times 1}$ and the beamforming matrix as $\mathbf{W} \in \mathbb{C}^{MN \times K}$, the signals received at the users and the eavesdroppers can be respectively expressed as
\begin{subequations}\label{eq_signal}
	\begin{align}
		&\mathbf{y}_{\mathrm{u}} = \mathbf{H}(\tilde{\mathbf{t}})^{\mathrm{H}} \mathbf{W} \mathbf{s} + \mathbf{n}_{\mathrm{u}},\\
		&\mathbf{y}_{\mathrm{e}} = \mathbf{G}(\tilde{\mathbf{t}})^{\mathrm{H}} \mathbf{W} \mathbf{s} + \mathbf{n}_{\mathrm{e}},
	\end{align}
\end{subequations}
where $\mathbf{n}_{\mathrm{u}}$ and $\mathbf{n}_{\mathrm{e}}$ are the complex Gaussian noise vectors with average power $\sigma^{2}$.
To guarantee ultra secure and covert transmission, the transmit beamforming matrix is designed based on the ZF criterion (assuming the instantaneous channels/locations of the users and eavesdroppers are known at the central aircraft for given $\tilde{\mathbf{t}}$) to completely eliminate the signal leakage to eavesdroppers, which is given by\footnote{In practice, the locations of users and eavesdroppers can be acquired using airborne radar. Furthermore, accurate radar sensing and localization enabled by the proposed ToMA array represent a promising direction for future research.}
\begin{equation}\label{eq_ZF}
	\begin{aligned}
		&\mathbf{W}(\tilde{\mathbf{t}}) = \sqrt{P} \frac{\bar{\mathbf{W}}(\tilde{\mathbf{t}})}{\left\|\bar{\mathbf{W}}(\tilde{\mathbf{t}})\right\|_{\mathrm{F}}},\\
		&\bar{\mathbf{W}}(\tilde{\mathbf{t}}) = \left[[\mathbf{H}(\tilde{\mathbf{t}}), \mathbf{G}(\tilde{\mathbf{t}})] ([\mathbf{H}(\tilde{\mathbf{t}}), \mathbf{G}(\tilde{\mathbf{t}})]^{\mathrm{H}} [\mathbf{H}(\tilde{\mathbf{t}}), \mathbf{G}(\tilde{\mathbf{t}})])^{-1}\right]_{:, 1:K},
	\end{aligned}
\end{equation}
where $P$ is the total transmit power. It is worth noting that the ZF beamforming in \eqref{eq_ZF} can not only realize ultra low received signal power at all eavesdroppers, but also ensure the same received signal power at all users. Thus, the corresponding achievable rate of each user in terms of bits-per-second-per-Hertz (bps/Hz) is given by
\begin{equation}\label{eq_rate}
	r(\tilde{\mathbf{t}}) = \log_{2}\Big{(}1 + \frac{P}{\left\|\bar{\mathbf{W}}(\tilde{\mathbf{t}})\right\|_{\mathrm{F}}^{2}\sigma^{2}} \Big{)}, 
\end{equation}
which guarantees the user fairness.

In practice, due to the random movement of the users and eavesdroppers, their channels may vary over time. In this paper, we aim to maximize the ergodic achievable rate of the users by optimizing the APV of the ToMA array based on any given statistical distribution of the locations of users and eavesdroppers, while the ZF beamforming is then applied to cater to time-varying locations of users/eavesdroppers with the APV fixed at its optimized solution. The optimization problem is thus formulated as
\begin{subequations}\label{eq_problem}
	\begin{align}
		\mathop{\max} \limits_{\tilde{\mathbf{t}}} ~ &\mathbb{E}\{r(\tilde{\mathbf{t}})\} \label{eq_problem_a}\\
		\mathrm{s.t.}~~  &\left\|\mathbf{t}_{m}\right\|=L,~1 \leq m \leq M,\label{eq_problem_b}\\
		&\left\|\mathbf{t}_{m} - \mathbf{t}_{\hat{m}}\right\| \geq D,~1 \leq m \neq \hat{m} \leq M,\label{eq_problem_c}
	\end{align}
\end{subequations}
where the expectation in \eqref{eq_problem_a} is conducted over random locations of the users and eavesdroppers as well as their associated wireless channels. Constraint \eqref{eq_problem_b} ensures that each towed ULA cable remains in its fully extended state, and constraint \eqref{eq_problem_c} ensures the minimum distance between any two towing drones to avoid collision. Problem \eqref{eq_problem} is challenging to solve optimally because the objective function does not have a closed form and the constraints are non-convex. To attain insights on antenna position optimization, we first conduct performance analysis for some simple cases in Section III. Then, a numerical algorithm for solving problem \eqref{eq_problem} is presented in the general case in Section IV.

\section{Performance Analysis}
In this section, we consider the simple case of \emph{a single user (i.e., $K=1$) and a single eavesdropper (i.e., $I=1$)} to demonstrate the fundamental performance advantage of the proposed ToMA array. As such, the ZF beamforming matrix in \eqref{eq_ZF} is recast into
\begin{equation}\label{eq_ZF2}
	\begin{aligned}
		&\mathbf{w}(\tilde{\mathbf{t}}) = \sqrt{P} \frac{\bar{\mathbf{w}}(\tilde{\mathbf{t}})}{\left\|\bar{\mathbf{w}}(\tilde{\mathbf{t}})\right\|},\\
		&\bar{\mathbf{w}}(\tilde{\mathbf{t}}) = \frac{\left(\mathbf{I}_{MN} - 	\frac{\mathbf{g}(\tilde{\mathbf{t}})\mathbf{g}(\tilde{\mathbf{t}})^{\mathrm{H}}}{\|\mathbf{g}(\tilde{\mathbf{t}})\|^{2}}\right) \mathbf{h}(\tilde{\mathbf{t}})}
		{\mathbf{h}(\tilde{\mathbf{t}})^{\mathrm{H}}\left(\mathbf{I}_{MN} - 	\frac{\mathbf{g}(\tilde{\mathbf{t}})\mathbf{g}(\tilde{\mathbf{t}})^{\mathrm{H}}}{\|\mathbf{g}(\tilde{\mathbf{t}})\|^{2}}\right) \mathbf{h}(\tilde{\mathbf{t}})},
	\end{aligned}
\end{equation}
where $\mathbf{h}(\tilde{\mathbf{t}})$ and $\mathbf{g}(\tilde{\mathbf{t}})$ respectively denote the channel vectors for the user and eavesdropper in \eqref{eq_channel}, with the subscript index omitted. The achievable rate in \eqref{eq_rate} is thus simplified as $r(\tilde{\mathbf{t}}) = \log_{2}(1 + \frac{P}{\left\|\bar{\mathbf{w}}(\tilde{\mathbf{t}})\right\|^{2}\sigma^{2}} )$. We then focus on maximizing the achievable rate for any given deterministic positions of the user and the eavesdropper, i.e., 
\begin{subequations}\label{eq_problem2}
	\begin{align}
		\mathop{\max} \limits_{\tilde{\mathbf{t}}} ~ &r(\tilde{\mathbf{t}}) = \log_{2}\Big{(}1 + \frac{P}{\left\|\bar{\mathbf{w}}(\tilde{\mathbf{t}})\right\|^{2}\sigma^{2}}\Big{)} \label{eq_problem2_a}\\
		\mathrm{s.t.}~~  &\eqref{eq_problem_b},~\eqref{eq_problem_c}.
	\end{align}
\end{subequations}
It can be observed that maximizing $r(\tilde{\mathbf{t}})$ is equivalently to maximizing $1/\|\bar{\mathbf{w}}(\tilde{\mathbf{t}})\|^{2}$. Substituting \eqref{eq_ZF2} and \eqref{eq_channel} into $\|\bar{\mathbf{w}}(\tilde{\mathbf{t}})\|^{2}$, we have
\begin{equation}\label{eq_ZF_power}
	\begin{aligned}
		\frac{1}{\|\bar{\mathbf{w}}(\tilde{\mathbf{t}})\|^{2}} &=  \mathbf{h}(\tilde{\mathbf{t}})^{\mathrm{H}} \left(\mathbf{I}_{MN} - \frac{\mathbf{g}(\tilde{\mathbf{t}})\mathbf{g}(\tilde{\mathbf{t}})^{\mathrm{H}}}{\|\mathbf{g}(\tilde{\mathbf{t}})\|^{2}}\right) \mathbf{h}(\tilde{\mathbf{t}}) \\
		&= \|\mathbf{h}(\tilde{\mathbf{t}})\|^{2} - \frac{|\mathbf{h}(\tilde{\mathbf{t}})^{\mathrm{H}} \mathbf{g}(\tilde{\mathbf{t}})|^{2}}{\|\mathbf{g}(\tilde{\mathbf{t}})\|^{2}}\\
		&=|\alpha|^{2}MN-\frac{|\alpha|^{2}|\mathbf{a}(\tilde{\mathbf{t}}, \mathbf{r}_{\mathrm{u}})^{\mathrm{H}} \mathbf{a}(\tilde{\mathbf{t}}, \mathbf{r}_{\mathrm{e}})|^{2}}{MN}.
	\end{aligned}
\end{equation}
Thus, minimizing $\|\bar{\mathbf{w}}(\tilde{\mathbf{t}})\|^{2}$ is equivalent to minimizing the correlation of the array response vectors, termed as \emph{array response correlation}, for the user and eavesdropper, i.e.,
\begin{subequations}\label{eq_problem3}
	\begin{align}
		\mathop{\min} \limits_{\tilde{\mathbf{t}}} ~ &|\mathbf{a}(\tilde{\mathbf{t}}, \mathbf{r}_{\mathrm{u}})^{\mathrm{H}} \mathbf{a}(\tilde{\mathbf{t}}, \mathbf{r}_{\mathrm{e}})| \label{eq_problem3_a}\\
		\mathrm{s.t.}~~  &\eqref{eq_problem_b},~\eqref{eq_problem_c}.
	\end{align}
\end{subequations}

Note that in the considered airborne communication scenario, the distance between the central aircraft and the user/eavesdropper is much larger than the ULA cable length $L$. Therefore, it is reasonably to assume that the user and eavesdropper are both located at beyond the Fresnel distance of the ToMA array. The distance between the transmit and receive antennas in \eqref{eq_array} can be well approximated by its second-order Taylor expansion as \cite{liu2023near}
\begin{equation}\label{eq_distance_far}{\small
		\begin{aligned}
			&\|\mathbf{t}_{m,n}-\mathbf{r}\| = \sqrt{\|\mathbf{t}_{m,n}\|^{2} - 2\mathbf{t}_{m,n}^{\mathrm{T}}\mathbf{r} + \|\mathbf{r}\|^{2}}\\
			\approx&\|\mathbf{r}\| - \hat{\mathbf{r}}^{\mathrm{T}}\mathbf{t}_{m,n}  + \frac{\|\mathbf{t}_{m,n}\|^{2} - (\hat{\mathbf{r}}^{\mathrm{T}}\mathbf{t}_{m,n})^{2}}{2\|\mathbf{r}\|},
	\end{aligned}}
\end{equation}
where $\hat{\mathbf{r}} \triangleq \frac{\mathbf{r}}{\|\mathbf{r}\|}$ denotes the normalized direction vector.

Next, we analyze the minimum array response correlation in \eqref{eq_problem3_a} under two different conditions, including 1) Far-field condition: The user and eavesdropper are located in the far-field region of the ToMA array; 2) Same-direction condition: The user and eavesdropper are located in the same direction to the center of the ToMA array. The corresponding maximum achievable rate can be obtained by directly substituting the optimal objective value in \eqref{eq_problem3_a} into \eqref{eq_ZF_power} and \eqref{eq_problem2}.

\subsection{Far-Field Condition}
Under the far-field condition with $\frac{\|\mathbf{t}_{m,n}\|^{2}}{\|\mathbf{r}\|} \ll \lambda$, $\forall m, n$, the third term in \eqref{eq_distance_far} is much smaller than the wavelength and thus can be neglected in the array response vector in \eqref{eq_array}.
Then, the array response correlation in \eqref{eq_problem3_a} can be expressed as
\begin{equation}\label{eq_corr_far}
	\begin{aligned}
		f_{\mathrm{ff}}(\tilde{\mathbf{t}}) \triangleq |\mathbf{a}(\tilde{\mathbf{t}}, \mathbf{r}_{\mathrm{u}})^{\mathrm{H}} \mathbf{a}(\tilde{\mathbf{t}}, \mathbf{r}_{\mathrm{e}})| 
		=\left|\sum \limits_{m=1}^{M} \sum \limits_{n=1}^{N} \e^{\jj \frac{2\pi}{\lambda} (\hat{\mathbf{r}}_{\mathrm{u}} - \hat{\mathbf{r}}_{\mathrm{e}})^{\mathrm{T}}\mathbf{t}_{m,n} } \right|.
	\end{aligned}
\end{equation}

\begin{theorem}\label{Theo_far}
	Under the far-field condition and a single towed ULA, i.e., $M=1$, the minimum array response correlation in \eqref{eq_corr_far} is given by 
	\begin{equation}\label{eq_corr_far_min}
		\mathop{\min} \limits_{\tilde{\mathbf{t}}} ~ f_{\mathrm{ff}}(\tilde{\mathbf{t}})=\left\{
		\begin{aligned}
			&0,~~~~~~~ \mathrm{if} ~\|\hat{\mathbf{r}}_{\mathrm{u}} - \hat{\mathbf{r}}_{\mathrm{e}}\|L \geq \lambda,\\
			&F_{\mathrm{ff}}(L),~ \mathrm{if}~ \|\hat{\mathbf{r}}_{\mathrm{u}} - \hat{\mathbf{r}}_{\mathrm{e}}\|L < \lambda,\\
			&N,~~~~~~ \mathrm{if} ~\|\hat{\mathbf{r}}_{\mathrm{u}} - \hat{\mathbf{r}}_{\mathrm{e}}\| = 0,
		\end{aligned}\right.
		\end{equation}
		with $F_{\mathrm{ff}}(L) \triangleq \left|\frac{\sin(\frac{\pi}{\lambda}\|\hat{\mathbf{r}}_{\mathrm{u}} - \hat{\mathbf{r}}_{\mathrm{e}}\|L)}{\sin(\frac{\pi}{\lambda}\frac{1}{N}\|\hat{\mathbf{r}}_{\mathrm{u}} - \hat{\mathbf{r}}_{\mathrm{e}}\|L)}\right|$ achieved when $\mathbf{t}_{1}=\pm L \frac{\hat{\mathbf{r}}_{\mathrm{u}} - \hat{\mathbf{r}}_{\mathrm{e}}}{\|\hat{\mathbf{r}}_{\mathrm{u}} - \hat{\mathbf{r}}_{\mathrm{e}}\|}$.
\end{theorem}
\begin{proof}
	See Appendix \ref{Appx_far}.
\end{proof}

According to the proof of Theorem~\ref{Theo_far}, if $\|\hat{\mathbf{r}}_{\mathrm{u}} - \hat{\mathbf{r}}_{\mathrm{e}}\| L < \lambda$, the minimum array response correlation $F_{\mathrm{ff}}(L)$ is achieved if and only if $\mathbf{t}_{1}$ is parallel to $(\hat{\mathbf{r}}_{\mathrm{u}} - \hat{\mathbf{r}}_{\mathrm{e}})$. In other words, the boresight of the ULA cable should be aligned with the angular bisector between the user and eavesdropper directions, thereby maximizing the effective array aperture for distinguishing between them. In this regime, the minimum correlation $F_{\mathrm{ff}}(L)$ is a decreasing function of the cable length $L$, as increasing $L$ leads to improved angular resolution. 

When the ULA cable is sufficiently long, i.e., $\|\hat{\mathbf{r}}_{\mathrm{u}} - \hat{\mathbf{r}}_{\mathrm{e}}\| L \geq \lambda$, the minimum array response correlation becomes zero. This can be achieved by gradually rotating the ULA cable away from the direction of $(\hat{\mathbf{r}}_{\mathrm{u}} - \hat{\mathbf{r}}_{\mathrm{e}})$ until the orthogonality condition $\sin\left( \frac{\pi}{\lambda} (\hat{\mathbf{r}}_{\mathrm{u}} - \hat{\mathbf{r}}_{\mathrm{e}})^{\mathrm{T}} \mathbf{t}_{1} \right) = 0$ is satisfied. In this case, the channel vectors of the user and eavesdropper become orthogonal, and the ZF beamforming reduces to maximum ratio transmission (MRT), i.e., $\bar{\mathbf{w}}(\tilde{\mathbf{t}}) = \mathbf{h}(\tilde{\mathbf{t}})/\|\mathbf{h}(\tilde{\mathbf{t}})\|^{2}$. This yields an upper bound on the achievable rate given by $\log_{2}(1 + \frac{|\alpha|^{2} N P}{\sigma^{2}})$. 

When the user and eavesdropper are aligned in the same direction relative to the ToMA array, i.e., $\|\hat{\mathbf{r}}_{\mathrm{u}} - \hat{\mathbf{r}}_{\mathrm{e}}\|^{2} = 0$, their channels become highly correlated, resulting in $\frac{1}{\|\bar{\mathbf{w}}(\tilde{\mathbf{t}})\|^{2}} = 0$. In this case, the achievable rate of the user drops to zero, indicating that secure transmission is not feasible under the far-field condition. However, as will be demonstrated later, extending the length of the towed ULA cable can effectively enlarge the near-field region of the ToMA array, distinguishing the user with the eavesdropper along the same direction.

\begin{figure}[t]
	\begin{center}
		\includegraphics[width=\figwidth cm]{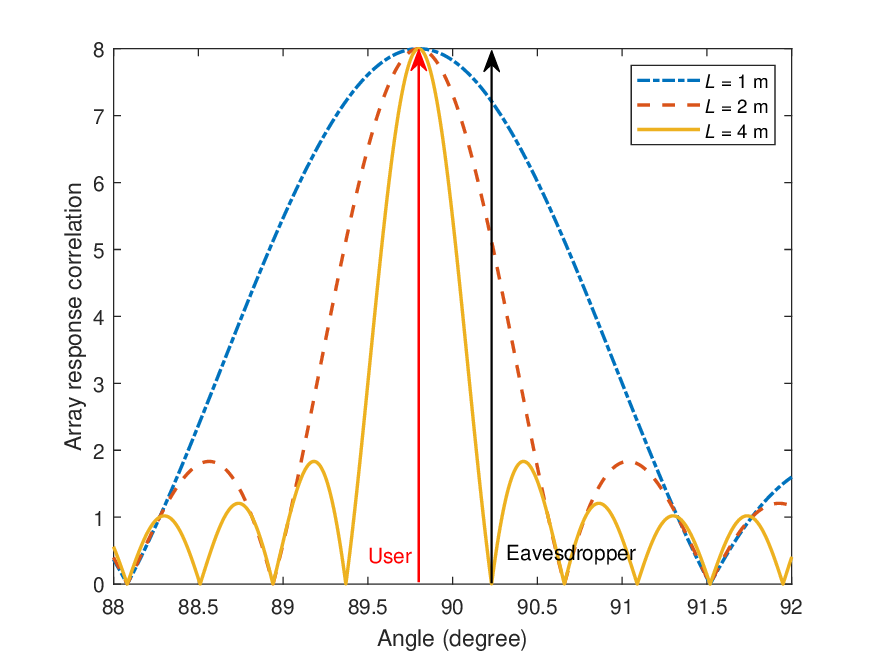}
		\caption{Array response correction between the user and eavesdropper under the far-field condition, with $M=1$, $N=8$, and $\lambda=0.03$ m.}
		\label{Fig_Beam_angular}
	\end{center}
\end{figure}

To illustrate the benefit of increasing the length of the towed ULA cable, Fig.~\ref{Fig_Beam_angular} presents the array response correlation between the user and the eavesdropper under far-field conditions. Specifically, we set \( \mathbf{t}_{1} = [L, 0, 0]^{\mathrm{T}} \), i.e., aligned along the \( x \)-axis. Both the user and the eavesdropper are assumed to lie on the \( x\text{-}O\text{-}z \) plane, with the user's angle relative to the \( x \)-axis fixed at \( 89.8^\circ \). As the eavesdropper's angle varies, the array response correlation exhibits different trends depending on the cable length. Notably, when the eavesdropper is located along angle \( 90.23^\circ \), the correlation for \( L = 4 \)~meters (m) approaches zero, whereas the correlation for \( L = 1 \)~m yields the highest value among all schemes. This demonstrates that, when the user and eavesdropper are closely spaced in the angular domain, increasing the ULA cable length significantly reduces the array response correlation, thereby enhancing the secure communication performance.

\begin{theorem}\label{Coro_far}
	Under the far-field condition and two towed ULAs, i.e., $M=2$, the minimum array response correlation in \eqref{eq_corr_far} is given by  
	\begin{equation}\label{eq_corr_far_min2}
		\mathop{\min} \limits_{\tilde{\mathbf{t}}} ~ f_{\mathrm{ff}}(\tilde{\mathbf{t}})=\left\{
			\begin{aligned}
				&0,~~~~~~~ \mathrm{if} ~2\frac{N+1}{N}\|\hat{\mathbf{r}}_{\mathrm{u}} - \hat{\mathbf{r}}_{\mathrm{e}}\|L \geq \lambda,\\
				&\tilde{F}_{\mathrm{ff}}(L),~ \mathrm{if}~ 2\frac{N+1}{N}\|\hat{\mathbf{r}}_{\mathrm{u}} - \hat{\mathbf{r}}_{\mathrm{e}}\|L < \lambda,\\
				&N,~~~~~~ \mathrm{if} ~\|\hat{\mathbf{r}}_{\mathrm{u}} - \hat{\mathbf{r}}_{\mathrm{e}}\| = 0,
		\end{aligned}\right.
	\end{equation}
	with $\tilde{F}_{\mathrm{ff}}(L) \triangleq 2\left|\cos(\frac{\pi}{\lambda}\frac{N+1}{N}\|\hat{\mathbf{r}}_{\mathrm{u}} - \hat{\mathbf{r}}_{\mathrm{e}}\|L)\right|F_{\mathrm{ff}}(L)$ achieved when $\mathbf{t}_{2}=-\mathbf{t}_{1}=\pm L \frac{\hat{\mathbf{r}}_{\mathrm{u}} - \hat{\mathbf{r}}_{\mathrm{e}}}{\|\hat{\mathbf{r}}_{\mathrm{u}} - \hat{\mathbf{r}}_{\mathrm{e}}\|}$.
\end{theorem}
\begin{proof}
	See Appendix \ref{Appx_far_coro}.
\end{proof}

According to the proof of Theorem~\ref{Coro_far}, an optimal solution always exists under the condition $\mathbf{t}_{1} = -\mathbf{t}_{2}$. This configuration, where two ULAs are collinear and symmetric, maximizes the array aperture along one spatial dimension. Similar observations to those in the single towed ULA case can then be drawn. Specifically, if $2\frac{N+1}{N}\|\hat{\mathbf{r}}_{\mathrm{u}} - \hat{\mathbf{r}}_{\mathrm{e}}\|L < \lambda$, the minimum correlation $\tilde{F}_{\mathrm{ff}}(L)$ is a decreasing function of the cable length $L$. In this case, the optimal APV should ensure that both $\mathbf{t}_{1}$ and $\mathbf{t}_{2}$ are parallel to $(\hat{\mathbf{r}}_{\mathrm{u}} - \hat{\mathbf{r}}_{\mathrm{e}})$, thereby maximizing the array aperture between the user and eavesdropper directions. Furthermore, when the ULA cables are sufficiently long, i.e., $2\frac{N+1}{N}\|\hat{\mathbf{r}}_{\mathrm{u}} - \hat{\mathbf{r}}_{\mathrm{e}}\|L \geq \lambda$, the array response vectors of the user and eavesdropper can become orthogonal via antenna position optimization. This leads to an upper bound on the achievable rate given by $\log_{2}(1 + \frac{2|\alpha|^{2}NP}{\sigma^{2}})$.

\subsection{Same-Direction Condition}
Under the same-direction condition with $\mathbf{r}_{\mathrm{u}} \propto \mathbf{r}_{\mathrm{e}}$, we denote the normalized direction vector as $\frac{\mathbf{r}_{\mathrm{u}}}{\|\mathbf{r}_{\mathrm{u}}\|}=\frac{\mathbf{r}_{\mathrm{e}}}{\|\mathbf{r}_{\mathrm{e}}\|}=\hat{\mathbf{r}}$. The array response correlation in \eqref{eq_problem3_a} can thus be simplified as
\begin{equation}\label{eq_corr_near}
	\begin{aligned}
		&f_{\mathrm{sd}}(\tilde{\mathbf{t}}) \triangleq |\mathbf{a}(\tilde{\mathbf{t}}, \mathbf{r}_{\mathrm{u}})^{\mathrm{H}} \mathbf{a}(\tilde{\mathbf{t}}, \mathbf{r}_{\mathrm{e}})| \\
		=&\left|\sum \limits_{m=1}^{M} \sum \limits_{n=1}^{N} \e^{\jj \frac{\pi}{\lambda} (\frac{1}{\|\mathbf{r}_{\mathrm{e}}\|}-\frac{1}{\|\mathbf{r}_{\mathrm{u}}\|})(\|\mathbf{t}_{m,n}\|^{2} - (\hat{\mathbf{r}}^{\mathrm{T}}\mathbf{t}_{m,n})^{2})} \right|.
	\end{aligned}
\end{equation}

For any given direction $\hat{\mathbf{r}}$, we define a virtual array whose element positions are given by $\mathbf{t}'_{m,n} \triangleq \mathbf{t}_{m,n} - \hat{\mathbf{r}}\hat{\mathbf{r}}^{\mathrm{T}}\mathbf{t}_{m,n}$, $1 \leq m \leq M$ and $1 \leq n \leq N$, representing the projection of $\mathbf{t}_{m,n}$ onto the plane orthogonal to $\hat{\mathbf{r}}$. It can be readily verified that $\|\mathbf{t}'_{m,n}\|^{2} - (\hat{\mathbf{r}}^{\mathrm{T}}\mathbf{t}'_{m,n})^{2} = \|\mathbf{t}_{m,n}\|^{2} - (\hat{\mathbf{r}}^{\mathrm{T}}\mathbf{t}_{m,n})^{2}$ always holds. It indicates that this virtual array yields the same array response correlation as in~\eqref{eq_corr_near}. In other words, different physical array configurations that share the same projected virtual array exhibit identical array correlation properties along direction $\hat{\mathbf{r}}$. Based on this finding, we provide the following theorem to demonstrate the minimum array response correlation.

\begin{theorem}\label{Theo_dir}
	Under the same-direction condition and a single towed ULA, i.e., $M=1$, the minimum array response correlation is given by
	\begin{equation}\label{eq_corr_near_min}
		\mathop{\min} \limits_{\tilde{\mathbf{t}}} ~ f_{\mathrm{sd}}(\tilde{\mathbf{t}})=\left\{
		\begin{aligned}
			&F_{\mathrm{sd}}(L),~ \mathrm{if}~ (\frac{1}{\|\mathbf{r}_{\mathrm{e}}\|}-\frac{1}{\|\mathbf{r}_{\mathrm{u}}\|})L^{2} < \lambda,\\
			&N,~~~~~~~ \mathrm{if} ~\|\hat{\mathbf{r}}_{\mathrm{u}}\| - \|\hat{\mathbf{r}}_{\mathrm{e}}\| = 0,
		\end{aligned}\right.
	\end{equation}
	with $F_{\mathrm{sd}}(L) \triangleq \left|\sum \limits_{n=1}^{N} \e^{\jj \frac{\pi}{\lambda} (\frac{1}{\|\mathbf{r}_{\mathrm{e}}\|}-\frac{1}{\|\mathbf{r}_{\mathrm{u}}\|})\frac{n^{2}}{N^{2}}L^{2}} \right|$ achieved when $\hat{\mathbf{r}}^{\mathrm{T}}\mathbf{t}_{1}=0$.
\end{theorem}
\begin{proof}
	See Appendix \ref{Appx_dir}.
\end{proof}

According to Theorem \ref{Theo_dir}, under the condition where the user and eavesdropper lie in the same direction with a short distance, the towed ULA should be oriented perpendicular to the wave direction, i.e., $\hat{\mathbf{r}}^{\mathrm{T}}\mathbf{t}_{1}=0$, in order to maximize the effective array aperture and thus improve distance separability. 
It is worth noting that in the near-field region, the minimum array response correlation generally cannot approach zero, as the phase of each element in \eqref{eq_corr_near} exhibits a nonlinear dependence on the antenna index \( n \). Under the condition \( (\frac{1}{\|\mathbf{r}_{\mathrm{e}}\|} - \frac{1}{\|\mathbf{r}_{\mathrm{u}}\|}) L^{2} < \lambda \), it can be readily shown that \( F_{\mathrm{sd}}(L) \) is a decreasing function of \( L \). This implies that increasing the ULA cable length can effectively enhance the secure communication performance.

To validate this observation, Fig.~\ref{Fig_Beam_boresight} plots the array response correlation between the user and the eavesdropper under the same-direction condition. Specifically, we set \( \mathbf{t}_{1} = [L, 0, 0]^{\mathrm{T}} \), aligned along the \( x \)-axis. Both the user and the eavesdropper are assumed to lie along the \( z \)-axis, with the user's distance from the origin fixed at 200 m. As the eavesdropper's distance varies, the array response correlation exhibits distinct trends for different cable lengths. In particular, when the eavesdropper's distance is larger than $50$ m, the correlation decreases with increasing cable length due to enhanced spatial separability.

\begin{figure}[t]
	\begin{center}
		\includegraphics[width=\figwidth cm]{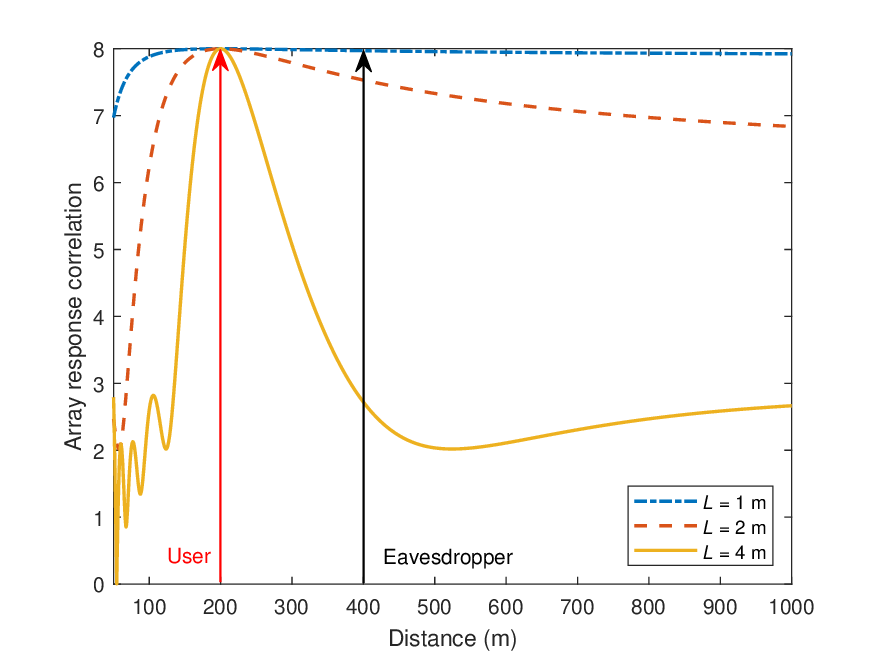}
		\caption{Array response correction between the user and eavesdropper under the same-direction condition, with $M=1$, $N=8$, and $\lambda=0.03$ m.}
		\label{Fig_Beam_boresight}
	\end{center}
\end{figure}

In the general case involving multiple towed ULAs, the presence of additional constraints on the minimum inter-cable distance makes it challenging to derive a unified optimal solution for the APV. Nevertheless, the system possesses higher DoFs for antenna position optimization, enabling the reduction of array response correlation between the user and eavesdropper. In this context, suboptimal APVs can be efficiently obtained through numerical algorithms. It is also foreseeable that increasing the length of ULA cables leads to reduced array response correlations, thereby enhancing secure communication performance.

\section{Optimization Algorithm}
In this section, we present a numerical algorithm to obtain a suboptimal solution for problem~\eqref{eq_problem}. The challenge arises from the objective function in~\eqref{eq_problem_a}, which involves the ergodic achievable rate and lacks a closed-form expression in general. To tackle this, we employ the Monte Carlo simulation method to approximate the objective function.
Specifically, given any statistical distribution of the users and eavesdroppers, we randomly generate their positions and corresponding channels over \( Q \) independent realizations. For sufficiently large $Q$, the expectation of the achievable rate in~\eqref{eq_problem_a} can be approximated by the sample average across all realizations, i.e.,
\begin{equation}\label{eq_obj}
	\mathbb{E}\{r(\tilde{\mathbf{t}})\} \approx \frac{1}{Q} \sum_{q=1}^{Q} r_{q}(\tilde{\mathbf{t}}) \triangleq \tilde{r}(\tilde{\mathbf{t}}),
\end{equation}
where \( r_{q}(\tilde{\mathbf{t}}) \) denotes the achievable rate under the \( q \)-th realization, \( 1 \leq q \leq Q \).

Next, we adopt the AO technique to iteratively optimize each ULA's position with the others being fixed. For the optimization of $\mathbf{t}_{m} \in \mathbb{R}^{3 \times 1}$, $1 \leq m \leq M$, we denote the achievable rate as its function given by $\tilde{r}_{m}(\mathbf{t}_{m})$. Then, the subproblem for optimizing $\mathbf{t}_{m}$ can be expressed as
\begin{subequations}\label{eq_problem_sub}
	\begin{align}
		\mathop{\max} \limits_{\mathbf{t}_{m}} ~ &\tilde{r}_{m}(\mathbf{t}_{m}) \label{eq_problem_sub_a}\\
		\mathrm{s.t.}~~  &\left\|\mathbf{t}_{m}\right\|=L,\label{eq_problem_sub_b}\\
		&\left\|\mathbf{t}_{m} - \mathbf{t}_{\hat{m}}\right\| \geq D,~\forall \hat{m} \neq m.\label{eq_problem_sub_c}
	\end{align}
\end{subequations} 

The primary challenge in solving problem \eqref{eq_problem_sub} lies in the intrinsic non-convexity of constraint \eqref{eq_problem_sub_b}, which renders conventional gradient-based methods in the Euclidean space inapplicable. Nonetheless, constraint \eqref{eq_problem_sub_b} defines a spherical manifold given by
\begin{equation}\label{eq_manifold}
	\mathcal{S} = \{\mathbf{t} \in \mathbb{R}^{3 \times 1} \mid \mathbf{t}^{\mathrm{T}}\mathbf{t} = L^{2} \}.
\end{equation}
This geometric structure allows the adoption of Riemannian manifold optimization techniques, ensuring that the iterates remain on $\mathcal{S}$ and thus automatically satisfy constraint \eqref{eq_problem_sub_b} throughout the optimization process.

To this end, we introduce some basic definitions in Riemannian manifold optimization \cite{yu2016alternating,boumal2023introduction}. The \emph{tangent space} of the spherical manifold $\mathcal{S}$ at point $\mathbf{t}$ is the set of all vectors in the Euclid space which are orthogonal to $\mathbf{t}$, given by
\begin{equation}\label{eq_tangent}
	T_{\mathbf{t}}\mathcal{S} = \{\mathbf{v} \in \mathbb{R}^{3 \times 1} \mid \mathbf{v}^{\mathrm{T}}\mathbf{t} = 0 \},
\end{equation}
wherein the tangent vectors specify all possible directions along which the point can move.

Among all the tangent vectors in $T_{\mathbf{t}}\mathcal{S}$, one of them is defined as the \emph{Riemanian gradient}, which represents the direction yielding the greatest increase of the objective function. The Riemannian gradient of function $\tilde{r}_{m}(\cdot)$ at point $\mathbf{t}$ is given by the orthogonal projection of the Euclidean gradient, denoted as $\nabla_\mathbf{t} \tilde{r}_{m}$, onto the tangent space $T_{\mathbf{t}}\mathcal{S}$, i.e.,
\begin{equation}\label{eq_grad}
	\mathrm{grad}_{\mathbf{t}} \tilde{r}_{m} = \nabla_\mathbf{t} \tilde{r}_{m} - \frac{\mathbf{t}\mathbf{t}^{\mathrm{T}}}{L^{2}} \nabla_\mathbf{t} \tilde{r}_{m},
\end{equation}
where the Euclidean gradient $\nabla_\mathbf{t} \tilde{r}_{m}$ can be numerically calculated according to the definition. 

Then, the gradient-based approach in traditional Euclidean space can be used in a similar way to solve the manifold optimization problem efficiently. Specifically, given the solution $\mathbf{t}^{(j)}$ in the $j$-th iteration, the search direction for maximizing the objective function in the tangent space is recursively defined as
\begin{equation}\label{eq_search_dir}
	\boldsymbol{\mu}^{(j)} = \mathrm{grad}_{\mathbf{t}^{(j)}} \tilde{r}_{m} +\kappa^{(j)} \mathrm{Transp}(\boldsymbol{\mu}^{(j-1)}),
\end{equation}
where $\kappa^{(j)}$ is the Polak-Ribi\`{e}re parameter to accelerate convergence. $\mathrm{Transp}(\cdot)$ denotes the \emph{transport} operation, which is a mapping between two vectors in different tangent spaces $T_{\mathbf{t}^{(j-1)}}\mathcal{S}$ and $T_{\mathbf{t}^{(j)}}\mathcal{S}$, given by
\begin{equation}\label{eq_transp}
	\begin{aligned}
		&\mathrm{Transp}_{\mathbf{t}^{(j-1)} \rightarrow \mathbf{t}^{(j)}}(\boldsymbol{\mu}^{(j-1)}) \triangleq T_{\mathbf{t}^{(j-1)}}\mathcal{S} \mapsto T_{\mathbf{t}^{(j)}}\mathcal{S}:\\
		&\boldsymbol{\mu}^{(j-1)} \mapsto \boldsymbol{\mu}^{(j-1)} - \frac{\mathbf{t}^{(j)}{\mathbf{t}^{(j)}}^{\mathrm{T}}}{L^{2}} \boldsymbol{\mu}^{(j-1)}.
	\end{aligned}
\end{equation}

Since the translation of the point along the tangent direction results in it moving out of the Riemannian manifold, we need to define the \emph{retraction} operation which maps a vector from the tangent space onto the manifold itself. The retraction of a tangent vector $\tau^{(j)}\boldsymbol{\mu}^{(j)}$ at point $\mathbf{t}^{(j)}$ can be expressed as
\begin{equation}\label{eq_ret}
	\begin{aligned}
		&\mathrm{Retr}_{\mathbf{t}^{(j)}}(\tau^{(j)}\boldsymbol{\mu}^{(j)}) \triangleq T_{\mathbf{t}^{(j)}}\mathcal{S} \mapsto \mathcal{S}:\\
		&\tau^{(j)}\boldsymbol{\mu}^{(j)} \mapsto L\frac{\mathbf{t}^{(j)} + \tau^{(j)}\boldsymbol{\mu}^{(j)}}{\|\mathbf{t}^{(j)} + \tau^{(j)}\boldsymbol{\mu}^{(j)}\|},
	\end{aligned}
\end{equation}
where $\tau^{(j)}$ denotes the step size obtained by backtracking line search. Specifically, the step size is initialized as a large positive value, $\tau^{(j)} = \tau_{\max}$. Then, we shrink the step size by a factor $\zeta \in (0,1)$ repeatedly, i.e., $\tau^{(j)} \leftarrow \zeta\tau^{(j)}$, until constraint \eqref{eq_problem_sub_c} and the Armijo–Goldstein condition are both satisfied, i.e.,
\begin{equation}\label{eq_Armijo}
	\begin{aligned}
		&\tilde{r}_{m}(\mathrm{Retr}_{\mathbf{t}^{(j)}}(\tau^{(j)}\boldsymbol{\mu}^{(j)})) \geq \tilde{r}_{m}(\mathbf{t}^{(j)}) + \xi \tau^{(j)} \left\|\mathrm{grad}_{\mathbf{t}^{(j)}} \tilde{r}_{m}\right\|,
	\end{aligned}
\end{equation}
where $\xi \in (0,1)$ is a predefined parameter to control the increasing speed of the objective function.

\begin{algorithm}[t]\small
	\caption{\small AO and Riemannian manifold optimization.}
	\label{Alg_Riemannian}
	\begin{algorithmic}[1]
		\REQUIRE ~$M$, $N$, $L$, $D$, $K$, $I$, $P$, $\sigma^{2}$, $\lambda$, $\tau_{\max}$, $\tau_{\min}$, $\zeta$, $\xi$, $\epsilon$, $J$, $T$.
		\STATE Initialize the APV $\tilde{\mathbf{t}}$.
		\STATE Generate $Q$ independent channel realizations.
		\STATE Obtain $\tilde{r}(\tilde{\mathbf{t}})$ according to \eqref{eq_obj}.
		\FOR   {$t=1:T$}
		\FOR   {$m=1:M$}
		\STATE Initialize $\mathbf{t}^{(1)} \leftarrow \mathbf{t}_{m}$ and $\boldsymbol{\mu}^{(0)} \leftarrow \mathbf{0}$.
		\FOR   {$j=1:J$}
		\STATE Calculate Riemannian gradient according to \eqref{eq_grad}.
		\STATE Set $\kappa^{(j)}=0.5/\|\mathrm{grad}_{\mathbf{t}^{(j)}} \tilde{r}_{m}\|$.
		\STATE Update search direction according to \eqref{eq_search_dir}.
		\STATE Initialize step size $\tau^{(j)} \leftarrow \tau_{\max}$.
		\WHILE {\eqref{eq_problem_sub_c} or \eqref{eq_Armijo} is not satisfied}
		\STATE Shrink the step size $\tau^{(j)} \leftarrow \zeta \tau^{(j)}$.
		\ENDWHILE
		\STATE Update $\mathbf{t}^{(j+1)} \leftarrow \mathrm{Retr}_{\mathbf{t}^{(j)}}(\tau^{(j)}\boldsymbol{\mu}^{(j)})$.
		\IF    {$\tau^{(j)}<\tau_{\min}$}
		\STATE Break.
		\ENDIF
		\STATE Update $\mathbf{t}_{m} \leftarrow \mathbf{t}^{(j+1)}$.
		\ENDFOR
		\ENDFOR
		\IF    {Increment of the objective function is below $\epsilon$}
		\STATE Break.
		\ENDIF
		\ENDFOR
		\RETURN $\tilde{\mathbf{t}}$.
	\end{algorithmic}
\end{algorithm}

The overall algorithm for solving problem \eqref{eq_problem} is summarized in Algorithm \ref{Alg_Riemannian}. First, we initialize the APV $\tilde{\mathbf{t}}$ using the hybrid placement scheme, which will be specified in Section V-A. Next, we generate $Q$ independent realizations of the user and eavesdropper locations, along with their corresponding channel vectors, which are then used to compute the ergodic achievable rate. In lines 4-25, we iteratively optimize the position of each towed ULA using the Riemannian manifold optimization method. In particular, $T$ denotes the maximum number of outer iterations for AO, while $J$ represents the maximum number of inner iterations for Riemannian manifold optimization. The Polak-Ribi\`{e}re parameter is set to $\kappa^{(j)}=0.5/\|\mathrm{grad}_{\mathbf{t}^{(j)}} \tilde{r}_{m}\|$ to balance between the Riemanian gradients over the current and previous iterations \cite{boumal2023introduction}. The inner iteration for Riemannian manifold optimization terminates if the step size $\tau^{(j)}$ is smaller than a predefined threshold $\tau_{\min}$. The outer iteration for AO terminates if the increment of the achievable rate between two consecutive iterations is below a predefined threshold $\epsilon$.

The computational complexity of Algorithm \eqref{eq_problem} is mainly owing to calculating the Riemannian gradient in line 8 and the backtracking line search in lines 12-13. Specifically, the computation of the ergodic achievable rate in \eqref{eq_obj} entails a complexity of $\mathcal{O}(QMN(K+I)^{2})$. Calculating the Euclidean gradient and the Riemannian gradient involves $(3M+1)$ times of computing the ergodic achievable rate. Thus, the corresponding computational complexity is given by $\mathcal{O}(QM^{2}N(K+I)^{2})$. Besides, denoting the maximum number of backtracking line search in lines 12-13 as $B=\log_{\zeta}\frac{\tau_{\min}}{\tau_{\max}}$, the corresponding computational complexity is $\mathcal{O}(BQMN(K+I)^{2})$. Given the maximum number of inner iterations, $J$, the maximum number of outer iterations, $T$, and the AO number, $M$, the total computational complexity of Algorithm \eqref{eq_problem} for solving \eqref{eq_problem} is thus given by $\mathcal{O}(TJ(M+B)QM^{2}N(K+I)^{2})$.

\section{Performance Evaluation}
In this section, we show the simulation results to evaluate performance of the proposed ToMA array enhanced secure communication system. 

\subsection{Simulation Setup and Benchmark Schemes}
In the simulation, the number of towed ULAs and the number of antennas per ULA are set to $M = 8$ and $N = 8$, respectively. Each towed ULA has a cable length of $L = 4$ m. To avoid collisions, the minimum inter-drone distance is set to $D = 0.5$ m. The numbers of users and eavesdroppers are both set to $K = 10$ and $I = 10$, respectively. The carrier frequency is chosen as $f_{\mathrm{c}} = 10$ GHz. The maximum transmit power of the ToMA array is $50$ dBm, and the average noise power at each user is $\sigma^2 = -90$ dBm. An LoS channel is assumed between the ToMA array and each user/eavesdropper, with path gains given by $\alpha_{k} = \frac{\lambda}{4\pi \|\mathbf{r}_{\mathrm{u},k}\|}$, $1 \leq k \leq K$, and $\beta_{i} = \frac{\lambda}{4\pi \|\mathbf{r}_{\mathrm{e},i}\|}$, $1 \leq i \leq I$. The other parameters used for Riemannian manifold optimization in Algorithm \ref{Alg_Riemannian} are set to $T=20$, $\epsilon=10^{-3}$, $J=100$, $\tau_{\max}=10^{-2}$, $\tau_{\min}=10^{-10}$, and $\zeta=0.5$. 

Unless otherwise specified, users and eavesdroppers are randomly distributed within three conical regions pointing forward ($x$-axis), leftward ($y$-axis), and downward ($-z$-axis) relative to the central aircraft. Each cone has a vertex angle of $10^{\circ}$, and the distance to the central aircraft ranges from $100$ m to $1000$ m. In addition to the proposed solution, the following benchmark schemes are defined for performance comparison.
\begin{itemize}
	\item \emph{Upper bound}: The performance upper bound is obtained by adopting MRT beamforming for each user and completely neglecting the multiuser interference and signal leakage, where the optimal power allocation is used to maximize the minimum achievable rate among users.
	\item \emph{Horizontal placement}: The towed ULAs are all placed on the horizontal plane (i.e., $x$-$O$-$y$ plane), with an equal intersection angle between any two adjacent ULAs. 
	\item \emph{Vertical placement}: The towed ULAs are all placed on the vertical plane (i.e., $x$-$O$-$z$ plane), with an equal intersection angle between any two adjacent ULAs. 
	\item \emph{Hybrid placement}: Half of the towed ULAs are placed on the horizontal plane and the remaining towed ULAs are placed on the vertical plane, with an equal intersection angle between any two adjacent ULAs on each plane.
	\item \emph{FPA-Dense UPA}: The antenna elements form a uniform planar array (UPA) of size $M \times N$, with the inter-antenna spacing given by $\lambda/2$.
	\item \emph{FPA-Sparse UPA}: The antenna elements form an UPA of size $M \times N$, with the inter-antenna spacing given by $2\lambda$.
\end{itemize}

\subsection{Simulation Results}
First, we evaluate the convergence behavior of the proposed Algorithm~\ref{Alg_Riemannian} in Fig.~\ref{Fig_Convergence}, where the total number of antennas is fixed to $MN=64$ and the total length of all ULA cables is fixed to $ML=32$ m. For different numbers of towed ULAs, $M$, the algorithm consistently exhibits fast convergence, typically within five (outer) iterations, and closely approaches the upper bound on the achievable rate. This demonstrates the effectiveness of both the proposed ToMA array architecture and the corresponding antenna position optimization method. Furthermore, the optimized antenna placements are illustrated in Fig.~\ref{Fig_APV}. It can be observed that the projections of the array geometry onto the $x$-$O$-$y$, $x$-$O$-$z$, and $y$-$O$-$z$ planes exhibit large and balanced apertures. Such a layout effectively reduces the array response correlation between users and eavesdroppers distributed in three conical areas, which aligns with the theoretical insights presented in Theorems~\ref{Theo_far} and~\ref{Theo_dir}.

\begin{figure}[t]
	\begin{center}
		\includegraphics[width=\figwidth cm]{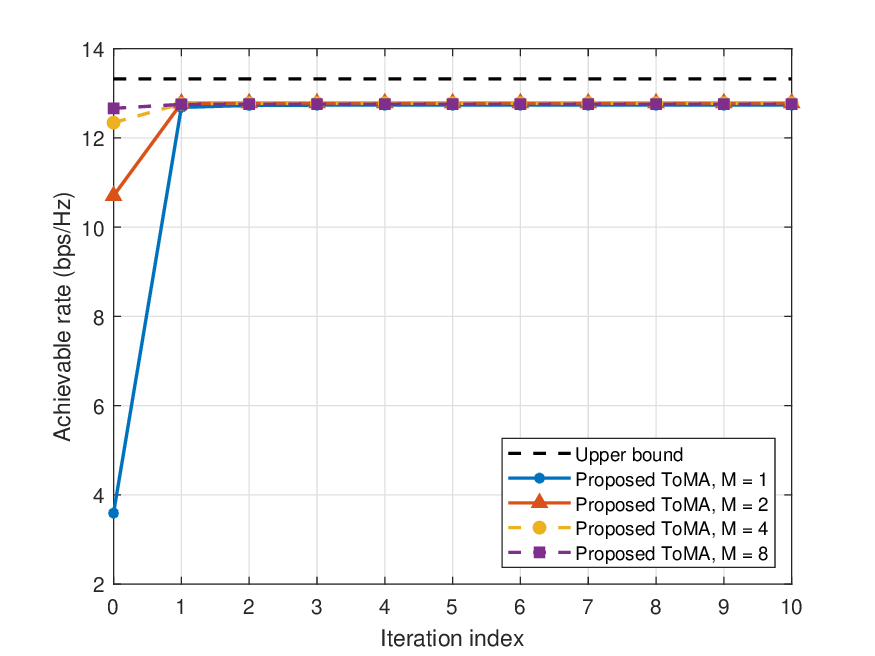}
		\caption{Convergence evaluation of the proposed Algorithm \ref{Alg_Riemannian}.}
		\label{Fig_Convergence}
	\end{center}
\end{figure}

\begin{figure*}[t]
	\centering
	\subfigure[$x$-$O$-$y$ plane]{\includegraphics[width=5.5 cm]{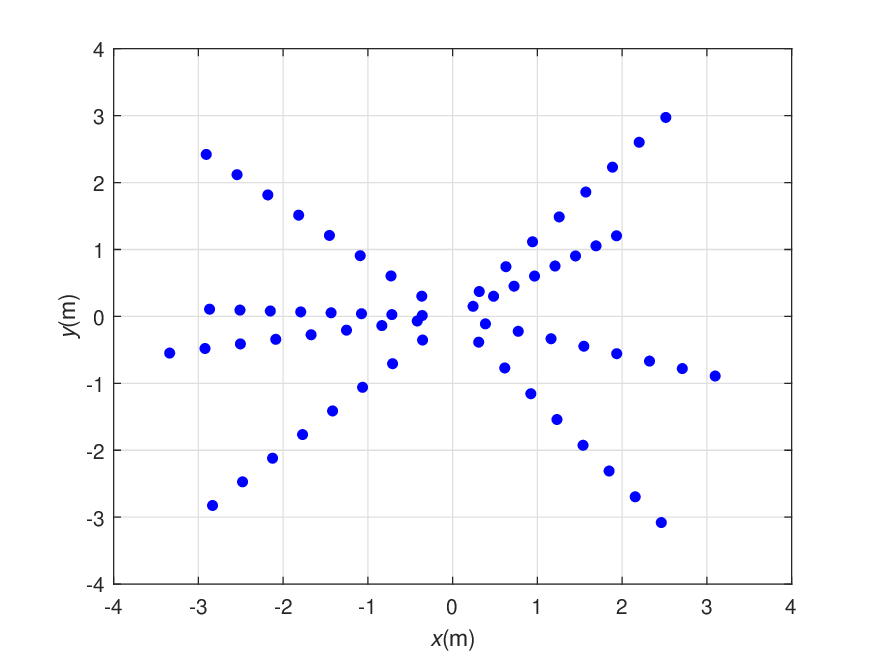} \label{Fig_APV_xoy}}
	\subfigure[$x$-$O$-$z$ plane]{\includegraphics[width=5.5 cm]{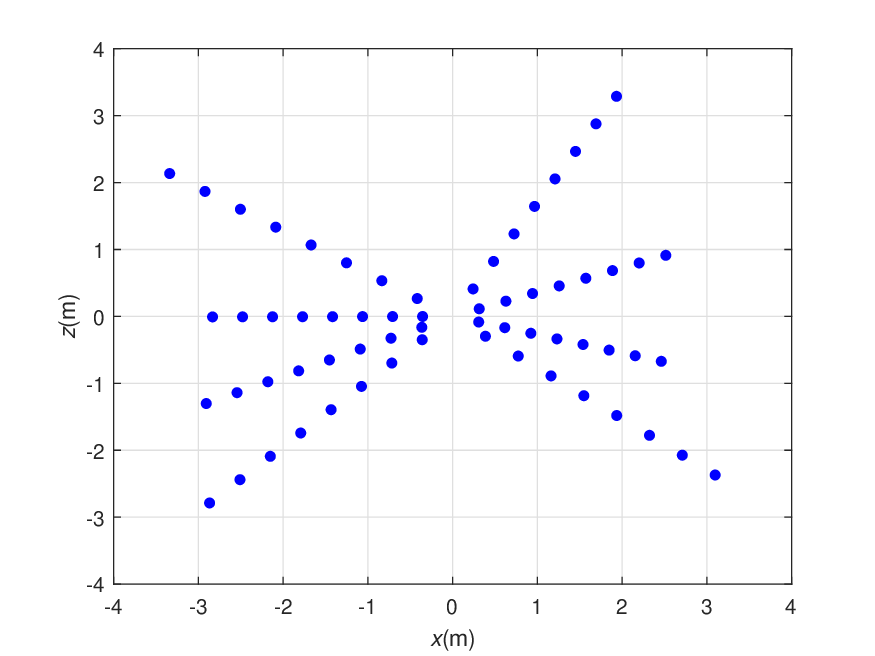} \label{Fig_APV_xoz}}
	\subfigure[$y$-$O$-$z$ plane]{\includegraphics[width=5.5 cm]{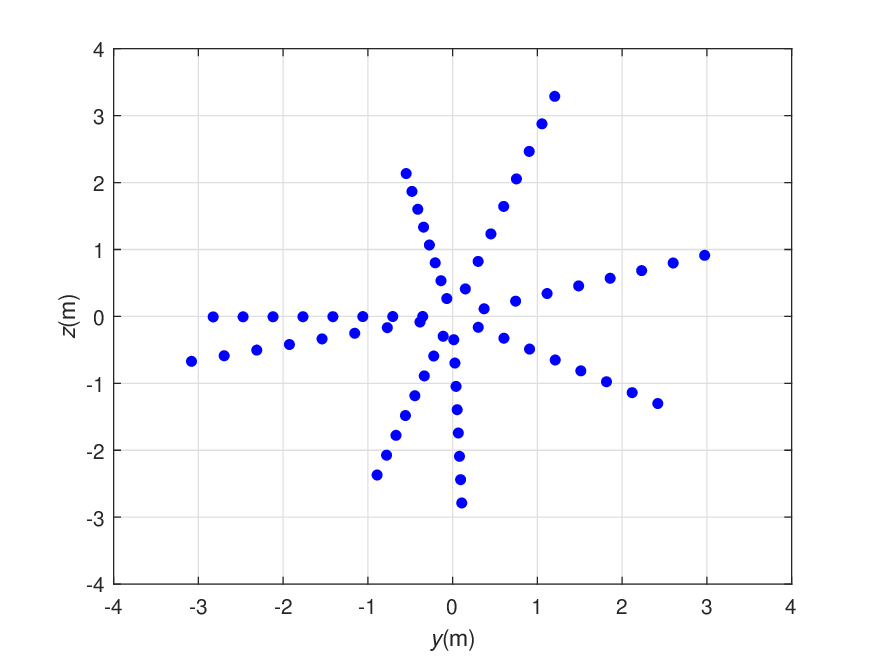} \label{Fig_APV_yoz}}
	\caption{The projection of optimized antenna positions on different planes.}
	\label{Fig_APV}
	\vspace{-15 pt}
\end{figure*}

\begin{figure}[t]
	\begin{center}
		\includegraphics[width=\figwidth cm]{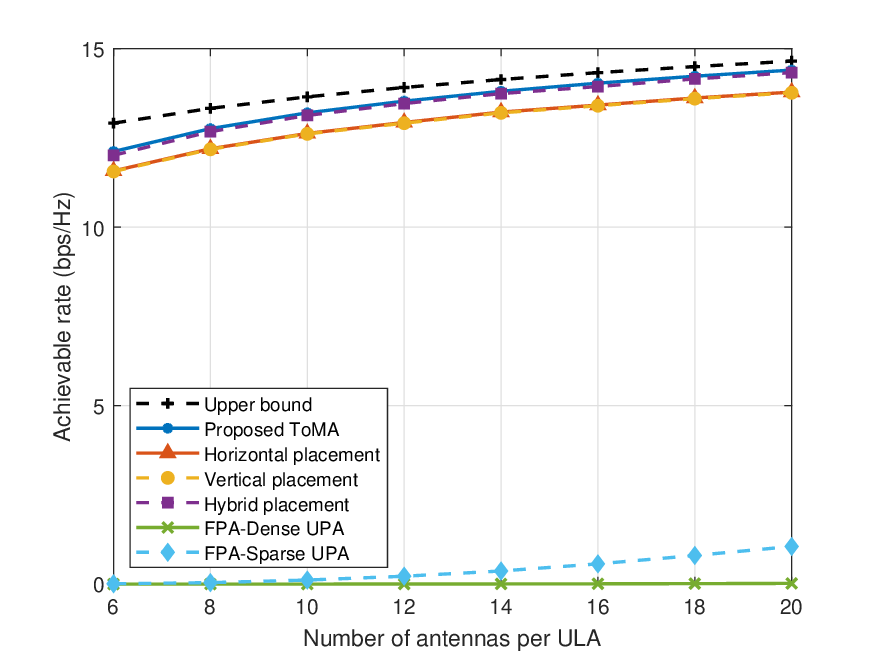}
		\caption{Performance comparison of the proposed and benchmark schemes versus the number of antennas per ULA, $N$, where the number of towed ULAs is fixed to $M=8$.}
		\label{Fig_Antenna_number}
	\end{center}
	\vspace{-12 pt}
\end{figure}
Next, we compare in Fig.~\ref{Fig_Antenna_number} the rate performance of the proposed and benchmark schemes as a function of the number of antennas per ULA, $N$. As shown, the proposed ToMA solution consistently achieves the highest rate performance among all schemes and closely approaches the upper bound, particularly for larger values of $N$. In contrast, conventional FPA schemes with either dense or sparse array layouts yield significantly lower achievable rates. This is because when the eavesdroppers are spatially close to the users, FPA arrays with limited apertures fail to effectively distinguish between them. By significantly expanding the array aperture, the proposed ToMA architecture enables ultra secure communication even under such challenging spatial proximity. Furthermore, the hybrid placement scheme performs comparably to the optimized antenna positioning approach. This is attributed to its ability to maintain relatively large and balanced apertures on the $x$-$O$-$y$, $x$-$O$-$z$, and $y$-$O$-$z$ planes, which helps reduce array response correlation between users and eavesdroppers in all three conical regions.

\begin{figure}[t]
	\begin{center}
		\includegraphics[width=\figwidth cm]{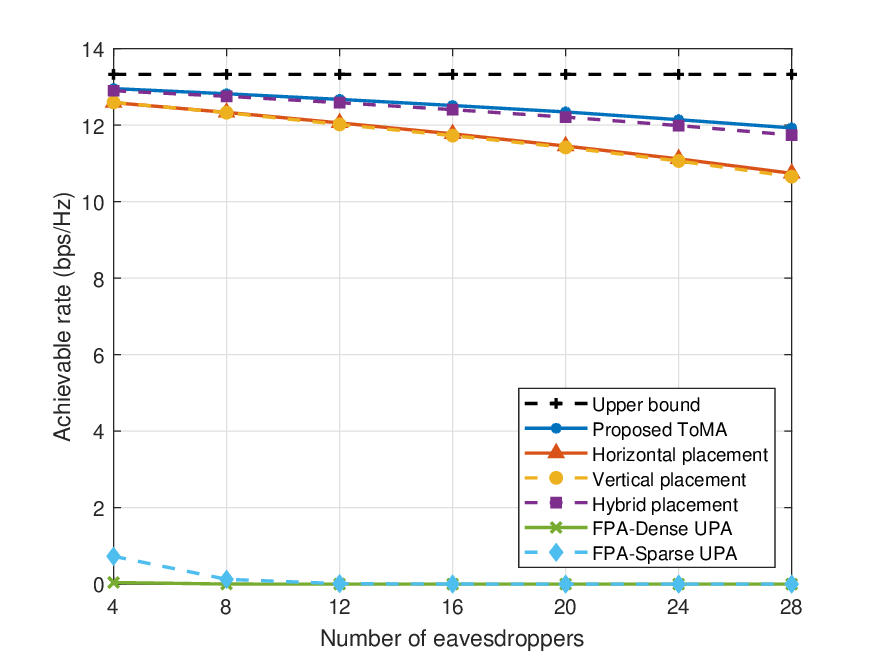}
		\caption{Performance comparison of the proposed and benchmark schemes versus the number of eavesdroppers, $I$.}
		\label{Fig_Eave_number}
	\end{center}
\end{figure}
Similar results can also be observed from Fig.~\ref{Fig_Eave_number}, where the proposed ToMA array achieves a much higher rate performance compared to the conventional FPA schemes. As the number of eavesdroppers increases, all schemes experience a certain degree of performance degradation. This is because the transmitter must suppress signal leakage to a larger number of potential eavesdroppers, which inevitably reduces the signal power received by legitimate users. Notably, the horizontal and vertical placement schemes suffer more severe performance losses, as they fail to maintain balanced and sufficiently large array apertures across different planes in the 3D space. This imbalance leads to higher channel correlations between the users and eavesdroppers, thereby compromising the secure communication performance.

\begin{figure}[t]
	\begin{center}
		\includegraphics[width=\figwidth cm]{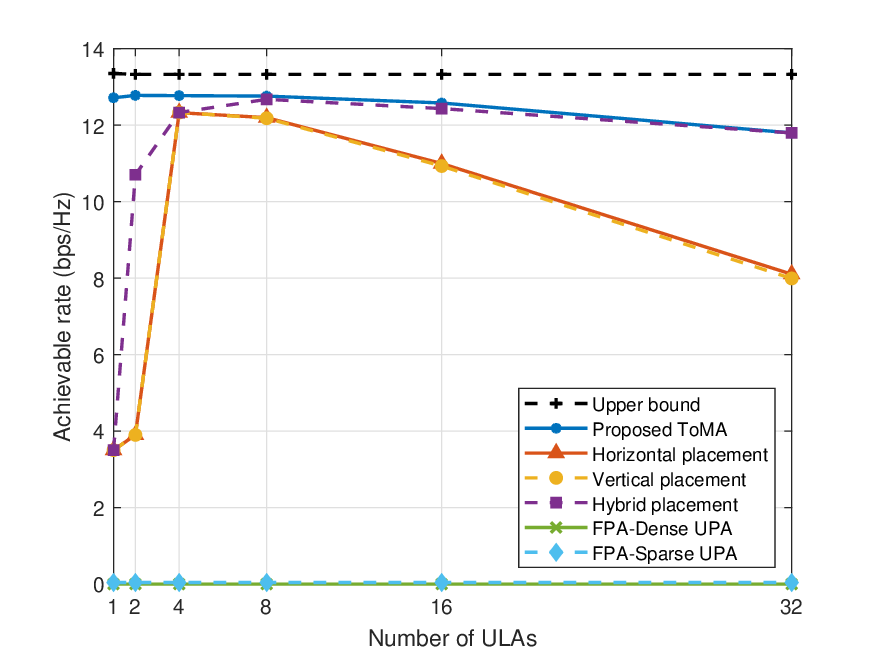}
		\caption{Performance comparison of the proposed and benchmark schemes versus the number of towed ULAs, $M$.}
		\label{Fig_ULA_number}
	\end{center}
\end{figure}
In Fig.~\ref{Fig_ULA_number}, we compare the performance of the proposed and benchmark schemes versus the number of towed ULAs, $M$, where the total number of antennas is fixed to $MN=64$ and the total cable length is fixed to $ML=32$~m. As a result, the performance upper bound remains constant across different values of $M$. For the proposed ToMA scheme, the achievable rate slightly improves as $M$ increases from $1$ to $2$, which is consistent with Theorem~\ref{Coro_far}, indicating that two optimally oriented towed ULAs outperform a single one. However, as $M$ continues to increase, the rate performance of the ToMA array begins to decline. This degradation arises because, under the fixed total cable length constraint $ML=32$~m, increasing $M$ reduces the length of each individual ULA, thereby shrinking the overall effective aperture of the array. Consequently, the channel correlation between users and eavesdroppers increases. Furthermore, when the number of towed ULAs is small (e.g., $M=1,2,4$), the horizontal, vertical, and hybrid placement schemes exhibit significant performance gaps compared to the proposed solution, highlighting the importance of optimizing antenna placement.

\begin{figure}[t]
	\centering
	\subfigure[Downward cone with $10^{\circ}$ vertex angle]{\includegraphics[width=\figwidth cm]{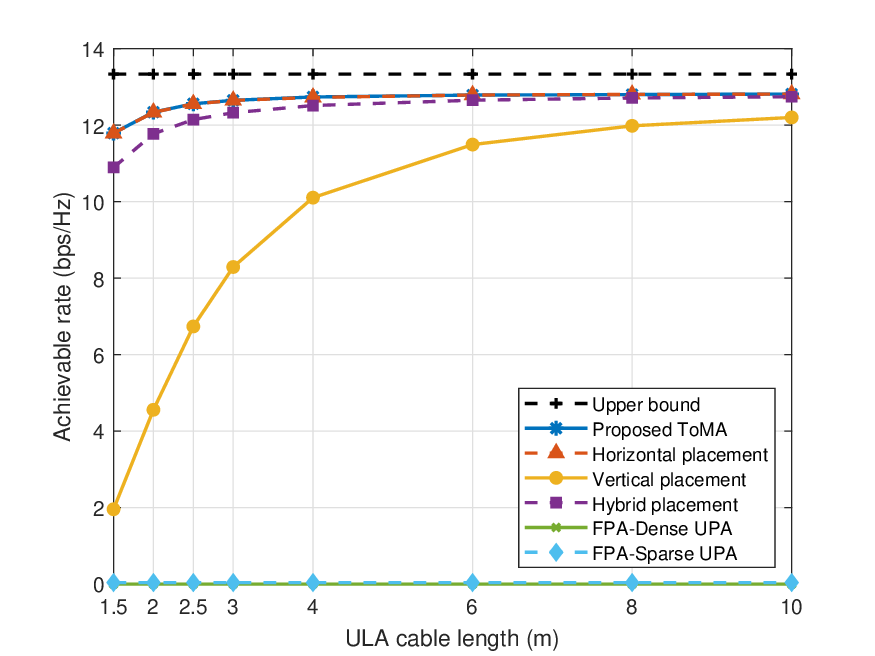} \label{Fig_Cable_length_downward10}}
	\subfigure[Leftward cone with $20^{\circ}$ vertex angle]{\includegraphics[width=\figwidth cm]{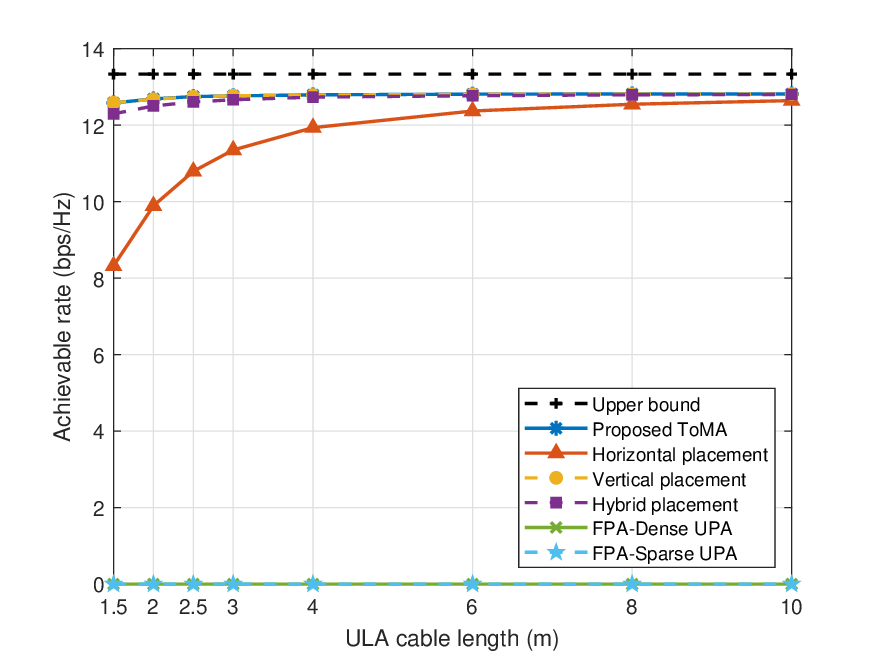} \label{Fig_Cable_length_leftward20}}
	\caption{Performance comparison of the proposed and benchmark schemes versus the ULA cable length, $L$.}
	\label{Fig_Cable_length}
\end{figure}
Fig.~\ref{Fig_Cable_length} illustrates the achievable rate performance of various schemes versus the length of each towed ULA cable, $L$, under different user/eavesdropper spatial distributions. In Fig.~\ref{Fig_Cable_length_downward10}, users and eavesdroppers are confined to a single downward conical area with a vertex angle of $10^{\circ}$. In this scenario, the horizontal placement scheme performs nearly identically to the proposed solution, since its array aperture is aligned with the dominant (downward) direction of the users and eavesdroppers. In contrast, the vertical placement scheme suffers significant performance degradation due to its limited effective aperture in that direction. In Fig.~\ref{Fig_Cable_length_leftward20}, the users and eavesdroppers are distributed in a single leftward conical region with a larger vertex angle of $20^{\circ}$. The performance trends are reversed, where the vertical placement scheme becomes nearly optimal, while the horizontal scheme performs poorly due to aperture misalignment. In both cases, the proposed ToMA array exhibits improved rate performance as the length of each towed ULA cable increases, benefiting from the enlarged array aperture. These results underscore the critical role of towing drones in enabling flexible aperture design for ultra-secure wireless communications.

\begin{figure}[t]
	\begin{center}
		\includegraphics[width=\figwidth cm]{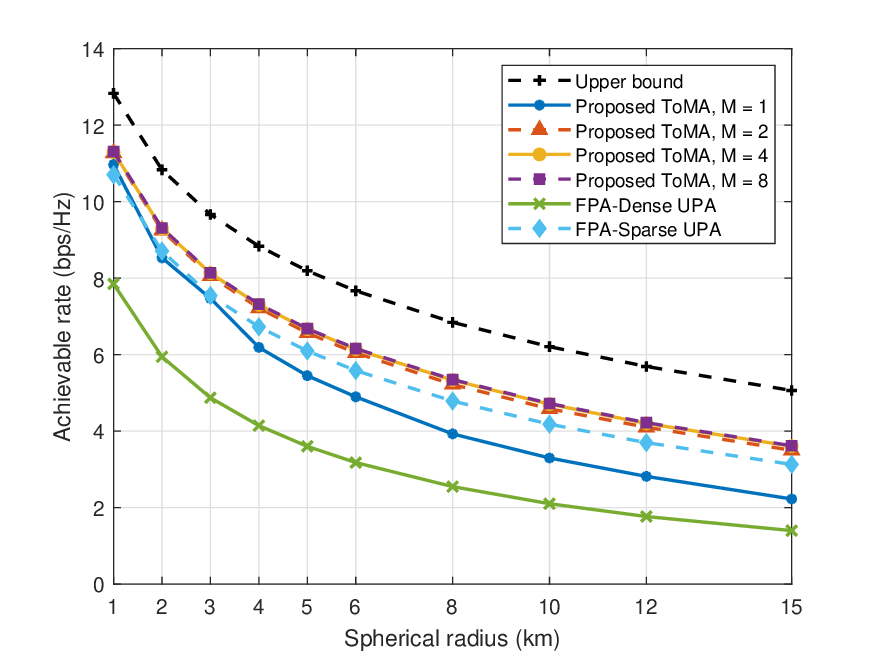}
		\caption{Performance comparison when the users and eavesdroppers are distributed on a spherical surface with varying radii.}
		\label{Fig_Spherical_radius}
	\end{center}
\end{figure}
Furthermore, in Fig.~\ref{Fig_Spherical_radius}, we compare the performance of different schemes when users and eavesdroppers are uniformly distributed on a spherical surface. As the radius of the spherical region increases, all schemes experience a decline in achievable rate due to increased path loss. We also evaluate the proposed ToMA scheme with varying numbers of towed ULAs, where the total number of antennas is fixed to $MN=64$ and the total cable length is fixed to $ML=32$~m for a fair comparison. When $M=1$, the ToMA array suffers from noticeable performance degradation, especially at larger spherical radii. This is attributed to the inability of a single ULA to maintain a large aperture in multiple directions within the 3D space. Specifically, the effective aperture of a single ULA is nearly zero along its radial direction. Moreover, due to its 1D structure, the array response correlation between any two directions with the same intersection angle to the ULA is always $MN$, severely degrading system performance. As the number of towed ULAs increases, the ToMA array has additional DoFs to balance effective apertures across different directions, thereby enhancing performance. Combining the insights from Figs.~\ref{Fig_Antenna_number} and \ref{Fig_Spherical_radius}, we observe that deploying two towed ULAs already yields satisfactory performance, while it may require frequent antenna reconfiguration to accommodate dynamic user and eavesdropper distributions. In contrast, the schemes employing more ULAs offer greater robustness under diverse spatial distributions of users and eavesdroppers, and even heuristic strategies such as hybrid placement can achieve high secure communication performance.

\begin{figure}[t]
	\begin{center}
		\includegraphics[width=\figwidth cm]{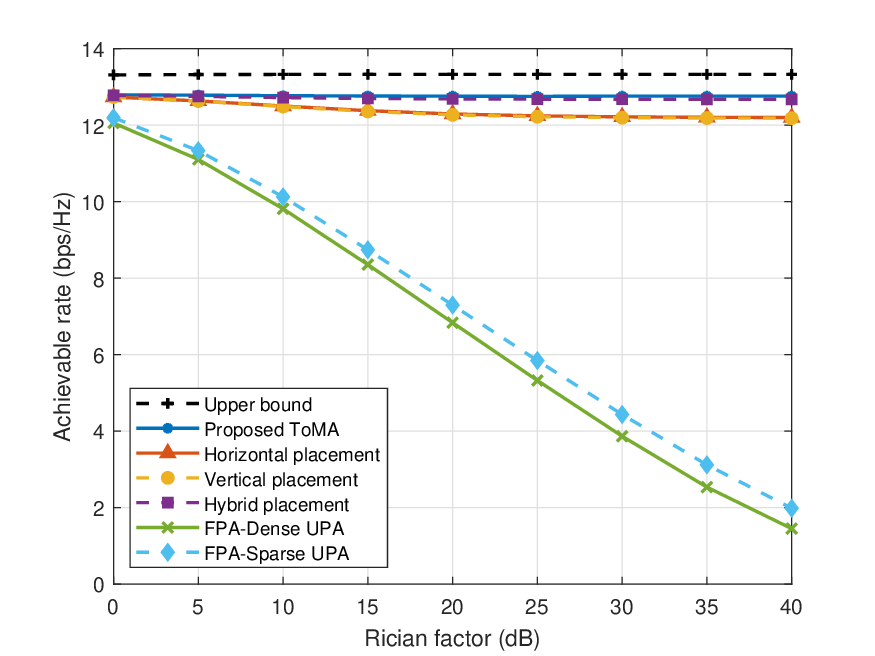}
		\caption{Performance comparison under Rician fading channels with varying Rician factors.}
		\label{Fig_Rician_factor}
	\end{center}
\end{figure}
Finally, we evaluate the performance of the proposed and benchmark schemes under Rician fading channels, as illustrated in Fig.~\ref{Fig_Rician_factor}. In addition to the LoS channel paths described in \eqref{eq_channel}, NLoS components are also present between the ToMA array and each user/eavesdropper. These NLoS components are modeled as independent and identically distributed (i.i.d.) random variables across different antenna positions, capturing the effects of random scattering in the environment. The canonical Rician factor is defined as the ratio between the average powers of the LoS and NLoS channel components. As shown in Fig.~\ref{Fig_Rician_factor}, the proposed ToMA scheme consistently achieves near-optimal performance across various Rician factor values, closely approaching the upper bound. In contrast, the benchmark FPA schemes exhibit significant performance degradation as the Rician factor increases. The reason is as follows. When the Rician factor is small, the channels resemble i.i.d. Rayleigh fading, where antenna placement has minimal impact on the channel statistics, resulting in similar performance between the ToMA and FPA schemes. However, as the Rician factor increases with the LoS components becoming more dominant, the ToMA scheme benefits from its larger effective array aperture, which significantly reduces channel correlation between users and eavesdroppers. Consequently, the performance gain of the proposed ToMA scheme becomes more pronounced in LoS-dominant environments, as commonly encountered in airborne communication systems.

\vspace{-0.2 cm}
\section{Conclusions}
This paper proposed a novel ToMA array architecture to enhance PLS in airborne communication systems. By employing drone-assisted deployment of multiple towed antenna subarrays, the ToMA array enables ultra-large effective apertures and dynamic 3D positioning beyond the physical constraints of conventional airborne platforms. These capabilities allow for highly flexible spatial beamforming and adaptive suppression of information leakage to eavesdroppers. We formulated the secure transmission problem under ZF beamforming and optimized the APV to maximize the users' ergodic achievable rate. Analytical results in the single-user and single-eavesdropper case revealed the impact of array geometry on channel correlation and secrecy performance. For the general multiuser scenario, a practical AO algorithm was developed on Riemannian manifolds to iteratively optimize the subarray positions. Simulation results demonstrated that the proposed ToMA array significantly outperforms FPA (i.e., dense/sparse UPA) systems, especially in challenging environments where eavesdroppers are spatially close to users under LoS-dominant channels. The results also highlight the importance of cable length and 3D geometry reconfigurability in achieving secure and robust airborne communications. Future work may explore the impact of directional antenna radiation patterns on position optimization and secure transmission performance. Besides, the effects of aerial platform jitter and air turbulence may result in uncertainty in antenna positions, which should be further modeled. In this context, robust signal processing techniques are required to cope with such challenging scenarios, representing an important direction for further research.

\vspace{-0.2 cm}
\appendices
\section{Proof of Theorem \ref{Theo_far}} \label{Appx_far}
For $M=1$, the array response correlation in \eqref{eq_corr_far} is simplified as
\begin{equation}\label{eq_corr_far_M1}
	\begin{aligned}
		f_{\mathrm{ff}}(\tilde{\mathbf{t}}) 
		=&\left|\sum \limits_{n=1}^{N} \e^{\jj \frac{2\pi}{\lambda} \frac{n}{N} (\hat{\mathbf{r}}_{\mathrm{u}} - \hat{\mathbf{r}}_{\mathrm{e}})^{\mathrm{T}}\mathbf{t}_{1} } \right|
		=\left|\frac{1 - \e^{\jj \frac{2\pi}{\lambda} (\hat{\mathbf{r}}_{\mathrm{u}} - \hat{\mathbf{r}}_{\mathrm{e}})^{\mathrm{T}}\mathbf{t}_{1} }}
		{1 - \e^{\jj \frac{2\pi}{\lambda} \frac{1}{N} (\hat{\mathbf{r}}_{\mathrm{u}} - \hat{\mathbf{r}}_{\mathrm{e}})^{\mathrm{T}}\mathbf{t}_{1} }}\right|\\
		=&\left|\frac{\sin(\frac{\pi}{\lambda}(\hat{\mathbf{r}}_{\mathrm{u}} - \hat{\mathbf{r}}_{\mathrm{e}})^{\mathrm{T}}\mathbf{t}_{1})}{\sin(\frac{\pi}{\lambda}\frac{1}{N}(\hat{\mathbf{r}}_{\mathrm{u}} - \hat{\mathbf{r}}_{\mathrm{e}})^{\mathrm{T}}\mathbf{t}_{1})}\right|.
	\end{aligned}
\end{equation}

If $\|\hat{\mathbf{r}}_{\mathrm{u}} - \hat{\mathbf{r}}_{\mathrm{e}}\| = 0$, the Dirichlet kernel function in \eqref{eq_corr_far_M1} is always equal to $N$. We next consider the case where $\|\hat{\mathbf{r}}_{\mathrm{u}} - \hat{\mathbf{r}}_{\mathrm{e}}\| \neq 0$. Under the constraint $\|\mathbf{t}_{1}\| = L$, the minimum and maximum values of the inner product $(\hat{\mathbf{r}}_{\mathrm{u}} - \hat{\mathbf{r}}_{\mathrm{e}})^{\mathrm{T}} \mathbf{t}_{1}$ are $-\|\hat{\mathbf{r}}_{\mathrm{u}} - \hat{\mathbf{r}}_{\mathrm{e}}\| L$ and $\|\hat{\mathbf{r}}_{\mathrm{u}} - \hat{\mathbf{r}}_{\mathrm{e}}\| L$, respectively. These extrema are attained when
$$
\mathbf{t}_{1}^{\min} = -L \frac{\hat{\mathbf{r}}_{\mathrm{u}} - \hat{\mathbf{r}}_{\mathrm{e}}}{\|\hat{\mathbf{r}}_{\mathrm{u}} - \hat{\mathbf{r}}_{\mathrm{e}}\|}, \quad
\mathbf{t}_{1}^{\max} = L \frac{\hat{\mathbf{r}}_{\mathrm{u}} - \hat{\mathbf{r}}_{\mathrm{e}}}{\|\hat{\mathbf{r}}_{\mathrm{u}} - \hat{\mathbf{r}}_{\mathrm{e}}\|}.
$$
As $\mathbf{t}_{1}$ continuously rotates from $\mathbf{t}_{1}^{\min}$ to $\mathbf{t}_{1}^{\max}$, the value of $(\hat{\mathbf{r}}_{\mathrm{u}} - \hat{\mathbf{r}}_{\mathrm{e}})^{\mathrm{T}} \mathbf{t}_{1}$ increases continuously from $-\|\hat{\mathbf{r}}_{\mathrm{u}} - \hat{\mathbf{r}}_{\mathrm{e}}\| L$ to $\|\hat{\mathbf{r}}_{\mathrm{u}} - \hat{\mathbf{r}}_{\mathrm{e}}\| L$.
Note that the Dirichlet kernel function in \eqref{eq_corr_far_M1} equals zero if and only if
$(\hat{\mathbf{r}}_{\mathrm{u}} - \hat{\mathbf{r}}_{\mathrm{e}})^{\mathrm{T}} \mathbf{t}_{1} = k\lambda,~\forall k \in \mathbb{Z} \setminus \{0\}$.
If $\|\hat{\mathbf{r}}_{\mathrm{u}} - \hat{\mathbf{r}}_{\mathrm{e}}\| L \geq \lambda$, then a $\mathbf{t}_{1}$ can always be found such that $(\hat{\mathbf{r}}_{\mathrm{u}} - \hat{\mathbf{r}}_{\mathrm{e}})^{\mathrm{T}} \mathbf{t}_{1} = \lambda$, yielding zero value of the Dirichlet kernel function. 

In contrast, if $\|\hat{\mathbf{r}}_{\mathrm{u}} - \hat{\mathbf{r}}_{\mathrm{e}}\| L < \lambda$, the Dirichlet kernel function increases as $(\hat{\mathbf{r}}_{\mathrm{u}} - \hat{\mathbf{r}}_{\mathrm{e}})^{\mathrm{T}} \mathbf{t}_{1}$ increases from $-\|\hat{\mathbf{r}}_{\mathrm{u}} - \hat{\mathbf{r}}_{\mathrm{e}}\| L$ to $0$, and then decreases as $(\hat{\mathbf{r}}_{\mathrm{u}} - \hat{\mathbf{r}}_{\mathrm{e}})^{\mathrm{T}} \mathbf{t}_{1}$ increases from $0$ to $\|\hat{\mathbf{r}}_{\mathrm{u}} - \hat{\mathbf{r}}_{\mathrm{e}}\| L$.
Therefore, the minimum value of the Dirichlet kernel function in \eqref{eq_corr_far_M1} is achieved at the two endpoints, i.e., when $(\hat{\mathbf{r}}_{\mathrm{u}} - \hat{\mathbf{r}}_{\mathrm{e}})^{\mathrm{T}} \mathbf{t}_{1} = \pm \|\hat{\mathbf{r}}_{\mathrm{u}} - \hat{\mathbf{r}}_{\mathrm{e}}\| L$, which yields
$$
F_{\mathrm{ff}}(L) = \left| \frac{\sin\left( \frac{\pi}{\lambda} \|\hat{\mathbf{r}}_{\mathrm{u}} - \hat{\mathbf{r}}_{\mathrm{e}}\| L \right)}{\sin\left( \frac{\pi}{\lambda} \frac{1}{N} \|\hat{\mathbf{r}}_{\mathrm{u}} - \hat{\mathbf{r}}_{\mathrm{e}}\| L \right)} \right|.
$$
This thus completes the proof.

\vspace{-0.3 cm}
\section{Proof of Theorem \ref{Coro_far}} \label{Appx_far_coro}
\begin{figure*}[!t]
	\centering 
	\begin{equation}\label{eq_corr_far_M2}
		\begin{aligned}
			f_{\mathrm{ff}}(\tilde{\mathbf{t}}) 
			=&\left|\sum \limits_{n=1}^{N} \e^{\jj \frac{2\pi}{\lambda} \frac{n}{N} (\hat{\mathbf{r}}_{\mathrm{u}} - \hat{\mathbf{r}}_{\mathrm{e}})^{\mathrm{T}}\mathbf{t}_{1} } + \sum \limits_{n=1}^{N} \e^{\jj \frac{2\pi}{\lambda} \frac{n}{N} (\hat{\mathbf{r}}_{\mathrm{u}} - \hat{\mathbf{r}}_{\mathrm{e}})^{\mathrm{T}}\mathbf{t}_{2} } \right|\\
			=&\left|\e^{\jj \frac{2\pi}{\lambda} \frac{1}{N} (\hat{\mathbf{r}}_{\mathrm{u}} - \hat{\mathbf{r}}_{\mathrm{e}})^{\mathrm{T}}\mathbf{t}_{1} } \frac{1 - \e^{\jj \frac{2\pi}{\lambda} (\hat{\mathbf{r}}_{\mathrm{u}} - \hat{\mathbf{r}}_{\mathrm{e}})^{\mathrm{T}}\mathbf{t}_{1} }}
			{1 - \e^{\jj \frac{2\pi}{\lambda} \frac{1}{N} (\hat{\mathbf{r}}_{\mathrm{u}} - \hat{\mathbf{r}}_{\mathrm{e}})^{\mathrm{T}}\mathbf{t}_{1} }}
			+ \e^{\jj \frac{2\pi}{\lambda} \frac{1}{N} (\hat{\mathbf{r}}_{\mathrm{u}} - \hat{\mathbf{r}}_{\mathrm{e}})^{\mathrm{T}}\mathbf{t}_{2} } \frac{1 - \e^{\jj \frac{2\pi}{\lambda} (\hat{\mathbf{r}}_{\mathrm{u}} - \hat{\mathbf{r}}_{\mathrm{e}})^{\mathrm{T}}\mathbf{t}_{2} }}
			{1 - \e^{\jj \frac{2\pi}{\lambda} \frac{1}{N} (\hat{\mathbf{r}}_{\mathrm{u}} - \hat{\mathbf{r}}_{\mathrm{e}})^{\mathrm{T}}\mathbf{t}_{2} }} \right|\\
			=&\left|\e^{\jj \frac{\pi}{\lambda} \frac{N+1}{N} (\hat{\mathbf{r}}_{\mathrm{u}} - \hat{\mathbf{r}}_{\mathrm{e}})^{\mathrm{T}}\mathbf{t}_{1}}
			\frac{\sin(\frac{\pi}{\lambda}(\hat{\mathbf{r}}_{\mathrm{u}} - \hat{\mathbf{r}}_{\mathrm{e}})^{\mathrm{T}}\mathbf{t}_{1})}{\sin(\frac{\pi}{\lambda}\frac{1}{N}(\hat{\mathbf{r}}_{\mathrm{u}} - \hat{\mathbf{r}}_{\mathrm{e}})^{\mathrm{T}}\mathbf{t}_{1})} 
			+ \e^{\jj \frac{\pi}{\lambda} \frac{N+1}{N}(\hat{\mathbf{r}}_{\mathrm{u}} - \hat{\mathbf{r}}_{\mathrm{e}})^{\mathrm{T}}\mathbf{t}_{2} }\frac{\sin(\frac{\pi}{\lambda}(\hat{\mathbf{r}}_{\mathrm{u}} - \hat{\mathbf{r}}_{\mathrm{e}})^{\mathrm{T}}\mathbf{t}_{2})}{\sin(\frac{\pi}{\lambda}\frac{1}{N}(\hat{\mathbf{r}}_{\mathrm{u}} - \hat{\mathbf{r}}_{\mathrm{e}})^{\mathrm{T}}\mathbf{t}_{2})} \right|.
		\end{aligned}
	\end{equation}
	\vspace*{8pt}
	\hrulefill
	\vspace{-12 pt}
\end{figure*}

For $M=2$, the array response correlation in \eqref{eq_corr_far} is given by \eqref{eq_corr_far_M2} shown at the top of the next page.
If $\|\hat{\mathbf{r}}_{\mathrm{u}} - \hat{\mathbf{r}}_{\mathrm{e}}\| = 0$, the array response correlation is always equal to $2N$. We next consider the case where $\|\hat{\mathbf{r}}_{\mathrm{u}} - \hat{\mathbf{r}}_{\mathrm{e}}\| \neq 0$. 
In particular, we consider to set $\mathbf{t}_{1}=-\mathbf{t}_{2}$, where the array response correlation in \eqref{eq_corr_far_M2} is simplified as
\begin{equation}\label{eq_corr_far_M2opt}{\small
	\begin{aligned}
		f_{\mathrm{ff}}(\tilde{\mathbf{t}}) 
		=\left|\cos{\Big(\frac{2\pi}{\lambda} \frac{N+1}{N} (\hat{\mathbf{r}}_{\mathrm{u}} - \hat{\mathbf{r}}_{\mathrm{e}})^{\mathrm{T}}\mathbf{t}_{1}\Big)}
		\frac{\sin(\frac{\pi}{\lambda}(\hat{\mathbf{r}}_{\mathrm{u}} - \hat{\mathbf{r}}_{\mathrm{e}})^{\mathrm{T}}\mathbf{t}_{1})}{\sin(\frac{\pi}{\lambda}\frac{1}{N}(\hat{\mathbf{r}}_{\mathrm{u}} - \hat{\mathbf{r}}_{\mathrm{e}})^{\mathrm{T}}\mathbf{t}_{1})} \right|.
	\end{aligned}}
\end{equation}
As demonstrated in Appendix~\ref{Appx_far}, when \( \mathbf{t}_{1} \) continuously rotates from \( \mathbf{t}_{1}^{\min} \) to \( \mathbf{t}_{1}^{\max} \), the value of \( (\hat{\mathbf{r}}_{\mathrm{u}} - \hat{\mathbf{r}}_{\mathrm{e}})^{\mathrm{T}} \mathbf{t}_{1} \) increases monotonically from \( -\|\hat{\mathbf{r}}_{\mathrm{u}} - \hat{\mathbf{r}}_{\mathrm{e}}\| L \) to \( \|\hat{\mathbf{r}}_{\mathrm{u}} - \hat{\mathbf{r}}_{\mathrm{e}}\| L \). 
Note that the array response correlation in \eqref{eq_corr_far_M2} becomes zero if and only if 
\[
(\hat{\mathbf{r}}_{\mathrm{u}} - \hat{\mathbf{r}}_{\mathrm{e}})^{\mathrm{T}} \mathbf{t}_{1} = k\lambda 
\quad \text{or} \quad 
\frac{N+1}{N} (\hat{\mathbf{r}}_{\mathrm{u}} - \hat{\mathbf{r}}_{\mathrm{e}})^{\mathrm{T}} \mathbf{t}_{1} = (k + \frac{1}{2} ) \lambda
\]
holds, $\forall k \in \mathbb{Z} \setminus \{0\}$. Under the condition \( 2\frac{N+1}{N} \|\hat{\mathbf{r}}_{\mathrm{u}} - \hat{\mathbf{r}}_{\mathrm{e}}\| L \geq \lambda \), there always exists a \( \mathbf{t}_{1} \) such that 
\[
\frac{N+1}{N} (\hat{\mathbf{r}}_{\mathrm{u}} - \hat{\mathbf{r}}_{\mathrm{e}})^{\mathrm{T}} \mathbf{t}_{1} = \frac{\lambda}{2},
\]
which leads to zero array response correlation.

If \( 2\frac{N+1}{N}\|\hat{\mathbf{r}}_{\mathrm{u}} - \hat{\mathbf{r}}_{\mathrm{e}}\|L < \lambda \), it is straightforward to verify that the two terms in \eqref{eq_corr_far_M2} are complex numbers with phases lying within the interval \( \left(-\frac{\pi}{2}, \frac{\pi}{2}\right) \). To minimize \eqref{eq_corr_far_M2}, the optimal APV should ensure that the phases 
\[
\frac{\pi}{\lambda} \frac{N+1}{N} (\hat{\mathbf{r}}_{\mathrm{u}} - \hat{\mathbf{r}}_{\mathrm{e}})^{\mathrm{T}}\mathbf{t}_{1}
\quad \text{and} \quad
\frac{\pi}{\lambda} \frac{N+1}{N} (\hat{\mathbf{r}}_{\mathrm{u}} - \hat{\mathbf{r}}_{\mathrm{e}})^{\mathrm{T}}\mathbf{t}_{2}
\]
have opposite signs. Otherwise, one can always set \( \mathbf{t}_{1} \leftarrow -\mathbf{t}_{1} \) or \( \mathbf{t}_{2} \leftarrow -\mathbf{t}_{2} \) to further reduce the array response correlation. Without loss of generality, we assume
\[
\frac{\pi}{\lambda} \frac{N+1}{N} (\hat{\mathbf{r}}_{\mathrm{u}} - \hat{\mathbf{r}}_{\mathrm{e}})^{\mathrm{T}}\mathbf{t}_{1} > 0
\quad \text{and} \quad
\frac{\pi}{\lambda} \frac{N+1}{N} (\hat{\mathbf{r}}_{\mathrm{u}} - \hat{\mathbf{r}}_{\mathrm{e}})^{\mathrm{T}}\mathbf{t}_{2} < 0.
\]

For any given \( \mathbf{t}_{1} \), it is easy to verify that \eqref{eq_corr_far_M2} monotonously decreases as \( (\hat{\mathbf{r}}_{\mathrm{u}} - \hat{\mathbf{r}}_{\mathrm{e}})^{\mathrm{T}}\mathbf{t}_{2} \) increases from 0 to \( -(\hat{\mathbf{r}}_{\mathrm{u}} - \hat{\mathbf{r}}_{\mathrm{e}})^{\mathrm{T}}\mathbf{t}_{1} \). Similarly, for any given \( \mathbf{t}_{2} \), \eqref{eq_corr_far_M2} monotonously increases as \( (\hat{\mathbf{r}}_{\mathrm{u}} - \hat{\mathbf{r}}_{\mathrm{e}})^{\mathrm{T}}\mathbf{t}_{1} \) increases from \( -(\hat{\mathbf{r}}_{\mathrm{u}} - \hat{\mathbf{r}}_{\mathrm{e}})^{\mathrm{T}}\mathbf{t}_{2} \) to 0. These observations imply that the optimal APV must satisfy \( \mathbf{t}_{2} = -\mathbf{t}_{1} \); otherwise, we can always set \( \mathbf{t}_{2} \leftarrow -\mathbf{t}_{1} \) or \( \mathbf{t}_{1} \leftarrow -\mathbf{t}_{2} \) to further reduce the correlation in \eqref{eq_corr_far_M2}.

Under the condition \( \mathbf{t}_{2} = -\mathbf{t}_{1} \), the expression in \eqref{eq_corr_far_M2} simplifies to \eqref{eq_corr_far_M2opt}. Given \( 2\frac{N+1}{N}\|\hat{\mathbf{r}}_{\mathrm{u}} - \hat{\mathbf{r}}_{\mathrm{e}}\|L < \lambda \), we can further verify that \eqref{eq_corr_far_M2opt} increases monotonically as $(\hat{\mathbf{r}}_{\mathrm{u}} - \hat{\mathbf{r}}_{\mathrm{e}})^{\mathrm{T}}\mathbf{t}_{1}$ increases from its minimum value \( -\|\hat{\mathbf{r}}_{\mathrm{u}} - \hat{\mathbf{r}}_{\mathrm{e}}\|L \) to $0$. Therefore, the minimum array response correlation is given by
\[
\tilde{F}_{\mathrm{ff}}(L) = 2\left|\cos\left(\frac{\pi}{\lambda} \frac{N+1}{N} \|\hat{\mathbf{r}}_{\mathrm{u}} - \hat{\mathbf{r}}_{\mathrm{e}}\|L \right)\right| F_{\mathrm{ff}}(L),
\]
which is achieved when
\[
\mathbf{t}_{2} = -\mathbf{t}_{1} = L \frac{\hat{\mathbf{r}}_{\mathrm{u}} - \hat{\mathbf{r}}_{\mathrm{e}}}{\|\hat{\mathbf{r}}_{\mathrm{u}} - \hat{\mathbf{r}}_{\mathrm{e}}\|}.
\]
This thus completes the proof.

\vspace{-0.2 cm}
\section{Proof of Theorem \ref{Theo_dir}} \label{Appx_dir}
Under the condition $M=1$, the array response correlation in \eqref{eq_corr_near} is simplified as
\begin{equation}\label{eq_corr_near_M1}
	\begin{aligned}
		f_{\mathrm{sd}}(\tilde{\mathbf{t}}) 
		=\left|\sum \limits_{n=1}^{N} \e^{\jj \frac{\pi}{\lambda} (\frac{1}{\|\mathbf{r}_{\mathrm{e}}\|}-\frac{1}{\|\mathbf{r}_{\mathrm{u}}\|})\frac{n^2}{N^2}(L^{2} - (\hat{\mathbf{r}}^{\mathrm{T}}\mathbf{t}_{1})^{2})} \right|.
	\end{aligned}
\end{equation}

If $\|\hat{\mathbf{r}}_{\mathrm{u}}\| - \|\hat{\mathbf{r}}_{\mathrm{e}}\| = 0$, it means that the user and eavesdropper have the same location, and thus the array response correlation in \eqref{eq_corr_near} is always equal to $N$. Note that as $\mathbf{t}_{1}$ rotates from $L\hat{\mathbf{r}}$ to a direction perpendicular to $\hat{\mathbf{r}}$, the quantity $(L^{2} - (\hat{\mathbf{r}}^{\mathrm{T}}\mathbf{t}_{1})^{2})$ continuously increases from its minimum value $0$ to its maximum value $L^{2}$. If $(\frac{1}{\|\mathbf{r}_{\mathrm{e}}\|}-\frac{1}{\|\mathbf{r}_{\mathrm{u}}\|})L^{2} < \lambda$, it is easy to verify that \eqref{eq_corr_near_M1} monotonously decreases as $(L^{2} - (\hat{\mathbf{r}}^{\mathrm{T}}\mathbf{t}_{1})^{2})$ increases from $0$ to $L^{2}$. Thus, the minimum array response correlation is given by 
$$F_{\mathrm{sd}}(L) = \left|\sum \limits_{n=1}^{N} \e^{\jj \frac{\pi}{\lambda} (\frac{1}{\|\mathbf{r}_{\mathrm{e}}\|}-\frac{1}{\|\mathbf{r}_{\mathrm{u}}\|})\frac{n^{2}}{N^{2}}L^{2}} \right|,$$
which is achieved when $\hat{\mathbf{r}}^{\mathrm{T}}\mathbf{t}_{1}=0$. This thus completes the proof.

\vspace{-0.3 cm}
\bibliographystyle{IEEEtran} 
\bibliography{IEEEabrv,ref_zhu}

\end{document}